\documentclass[12pt,a4paper]{amsart}
\usepackage{fullpage,hyperref,amssymb,pgfplots,tikz-cd,multicol,enumitem,xr}
\pgfplotsset{compat=1.18}
\hypersetup{colorlinks=true,linkcolor=blue,urlcolor=blue,citecolor=blue}
\title{Generalized Hofstadter functions G, H and beyond:
  numeration systems and discrepancy }

\author{Pierre Letouzey}
\address{Universit\'e Paris Cit\'e, CNRS, Inria, IRIF, F-75013 Paris, France}
\email{letouzey@irif.fr}

\theoremstyle{plain}
\newtheorem{theorem}{Theorem}
\newtheorem{corollary}[theorem]{Corollary}
\newtheorem{lemma}[theorem]{Lemma}
\newtheorem{proposition}[theorem]{Proposition}
\theoremstyle{definition}
\newtheorem{definition}[theorem]{Definition}

\numberwithin{equation}{section}
\numberwithin{theorem}{section}
\numberwithin{figure}{section}

\newcommand{\seqnum}[1]{\href{https://oeis.org/#1}{\rm \underline{#1}}}
\newcommand{\A}[2]{\ensuremath{A_{#1,#2}}}
\newcommand{\dmax}[2]{\ensuremath{\delta_{#1,#2}^+}}
\newcommand{\dmin}[2]{\ensuremath{\delta_{#1,#2}^-}}
\newcommand{\tF}{\tilde{F}}
\newcommand{\xmark}{\scalebox{0.85}[1]{\ensuremath{\times}}}
\newcommand{\modd}[2]{#1\ \mbox{\rm (mod}\ #2\mbox{\rm )}}
\DeclareMathOperator{\RE}{Re}

\DeclareMathOperator{\rank}{rank}
\DeclareMathOperator{\norm}{norm}

\begin{document}

\begin{abstract}
Hofstadter's G function is recursively defined via $G(0)=0$ and then
$G(n)=n-G(G(n-1))$. Following Hofstadter, a family $(F_k)$ of similar
functions is obtained by varying the number $k$ of nested recursive calls
in this equation. We study here some Fibonacci-like sequences that are
deeply connected with these functions $F_k$. In particular,
the Zeckendorf theorem can be adapted to provide digital
expansions via sums of terms of these sequences.
On these digital expansions, the functions
$F_k$ are acting as right shifts of the digits.
These Fibonacci-like sequences can be expressed in terms of zeros
of the polynomial
$X^k{-}X^{k-1}{-}1$.
Considering now the discrepancy of each
function $F_k$, i.e., the maximal distance between $F_k$ and its linear
equivalent, we retrieve the fact that this discrepancy is finite
exactly when $k \le 4$. Thanks to that, we solve two twenty-year-old
OEIS conjectures stating how close
the functions $F_3$ and $F_4$ are from
the integer parts of their linear equivalents.
Moreover we establish that $F_k$ can coincide exactly with such an integer
part only when $k\le 2$, while $F_k$ is almost additive exactly when $k \le 4$.
Finally, a nice fractal shape a la Rauzy has been encountered when
investigating the discrepancy of $F_3$. Almost all this article has
been formalized and verified in the Coq/Rocq proof assistant.
\end{abstract}

\maketitle

\section{introduction}

We continue here the exploration of a family of functions started
in a previous article~\cite{LetouzeyShuoSteiner2024}. The properties
studied here are largely independent from the content of the first
article, so this second article can be read without prior knowledge
of the first one. In particular the reader will find again below
the necessary initial definitions and basic properties. That being
said, we obviously recommend reading both. For instance, the previous
discussion on related works is still relevant here and will not be
duplicated. For this text, the major related works are due to Meek and Van
Rees~\cite{MeekVanRees84} for the first half and Dilcher~\cite{Dilcher1993}
for the second half. Actually, most of the current work has been done
before becoming aware of these two references. Despite an important
overlap, we decided to leave here all our content, even
the statements with prior proofs.
This document is hence mainly self-contained, while going slightly beyond
these two former references, via some new statements as well as a few
improved proofs. Moreover, almost all this work has been formalized
and verified in the Coq/Rocq proof assistant, so this article is
also an documentation of this formal development.

\subsection{The functions}

As in our previous article~\cite{LetouzeyShuoSteiner2024}, we consider
the following
function\footnote{In all this work, the
$n$-th exponent of a function denotes its $n$-th iterate,
for instance $F_2^2(n)$ is $F_2(F_2(n))$, not the square of
$F_2(n)$.}:
\begin{definition}For each integer $k\ge 1$,
the function $F_k$ is defined recursively by
\[
F_k:\, \mathbb{N} \to \mathbb{N}, \quad n \mapsto
\begin{cases}
0, & \text{if $n = 0$;} \\
n-F_k^k(n{-}1), & \text{otherwise.}
\end{cases}
\]
\end{definition}
This function $F_k$ is well defined since one may prove alongside that
$0 \le F_k(n) \le n$ for all $n \ge 0$. This family of functions is
due to Hofstadter \cite[Chapter~5]{GEB}.
In particular, $F_2$ is Hofstadter's
function~$G$, see OEIS entry~\seqnum{A5206}
\cite{OEIS,GEB,DowneyGriswold84,GaultClint}, known to satisfy
$G(n)=\lfloor (n{+}1)/\varphi \rfloor$ where $\varphi$ is the
golden ratio $(1{+}\sqrt{5})/2$. Similarly, $F_3$ is
Hofstadter's function~$H$, see OEIS~\seqnum{A5374} and also the recent
study by Shallit~\cite{Shallit2025}.
The generalization to
higher degrees of recursive nesting in the definition of~$F_k$
was already suggested by Hofstadter~\cite{GEB} and investigated by
Dilcher~\cite{Dilcher1993}.
To be complete, the OEIS database already includes
$F_4$ as~\seqnum{A5375} and
$F_5$ as~\seqnum{A5376} and
$F_6$ as~\seqnum{A100721}. On the other hand, we
choose to start this sequence with $F_1$ where only one recursive call
is done, leading to a function that can easily be shown to verify
$F_1(n)=\lfloor (n{+}1)/2 \rfloor=\lceil n/2 \rceil$.
Throughout this paper, we never consider the case $k=0$:
although the previous definition could be extended and give a
non-recursive function~$F_0$, this $F_0$
has too little in common with the other~$F_k$ functions to be
of much interest. Small values of the functions $F_1$ to~$F_5$ are
displayed in Figure~\ref{fig:plot}.

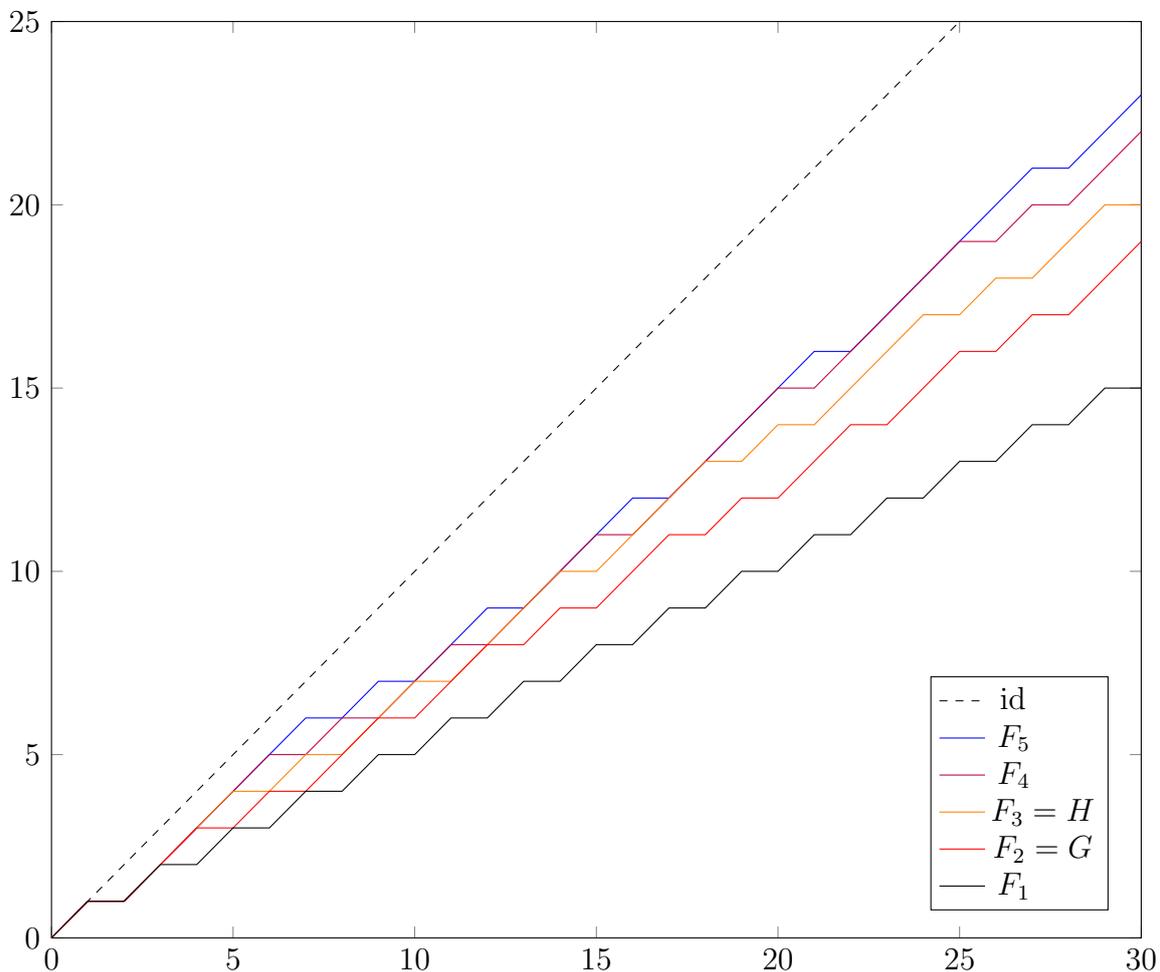
\begin{figure}[ht]
\pgfplotsset{width=\linewidth}
\begin{tikzpicture}[scale=1]
  \begin{axis}[
    xmin=0,xmax=30,ymin=0,ymax=25,samples=31,
    xtick = {0,5,...,30},
    ytick = {0,5,...,30},
    legend pos=south east
  ]
 \addplot+[mark=.,color=black,style=dashed,domain=0:30] {x};
 \addlegendentry{id \phantom{=G}}
    \addplot [mark=.,color=blue] coordinates {
(0, 0) (1, 1) (2, 1) (3, 2) (4, 3) (5, 4) (6, 5) (7, 6) (8, 6)
 (9, 7) (10, 7) (11, 8) (12, 9) (13, 9) (14, 10) (15, 11) (16, 12)
 (17, 12) (18, 13) (19, 14) (20, 15) (21, 16) (22, 16) (23, 17)
 (24, 18) (25, 19) (26, 20) (27, 21) (28, 21) (29, 22) (30, 23)};
 \addlegendentry{$F_5 \phantom{=G}$}
    \addplot [mark=.,color=purple] coordinates {
(0, 0) (1, 1) (2, 1) (3, 2) (4, 3) (5, 4) (6, 5) (7, 5) (8, 6)
 (9, 6) (10, 7) (11, 8) (12, 8) (13, 9) (14, 10) (15, 11) (16, 11)
 (17, 12) (18, 13) (19, 14) (20, 15) (21, 15) (22, 16) (23, 17)
 (24, 18) (25, 19) (26, 19) (27, 20) (28, 20) (29, 21) (30, 22)};
 \addlegendentry{$F_4 \phantom{=G}$}
    \addplot [mark=.,color=orange] coordinates {
(0, 0) (1, 1) (2, 1) (3, 2) (4, 3) (5, 4) (6, 4) (7, 5) (8, 5)
 (9, 6) (10, 7) (11, 7) (12, 8) (13, 9) (14, 10) (15, 10) (16, 11)
 (17, 12) (18, 13) (19, 13) (20, 14) (21, 14) (22, 15) (23, 16)
 (24, 17) (25, 17) (26, 18) (27, 18) (28, 19) (29, 20) (30, 20) };
 \addlegendentry{$F_3=H$}
 \addplot+[mark=.,color=red,style=solid,domain=0:30]
 {floor((x+1) * 0.618033988749894903)};
 \addlegendentry{$F_2=G$}
 \addplot+[mark=.,color=black,style=solid,domain=0:30] {ceil(x/2)};
 \addlegendentry{$F_1 \phantom{=G}$}
 \end{axis}
\end{tikzpicture}
\caption{Plotting $F_1,F_2,\ldots,F_5$.}
\label{fig:plot}
\end{figure}

\subsection{Summary of previous results}

The main result of our previous article~\cite{LetouzeyShuoSteiner2024}
is the Theorem~\ref*{article1:t:2}, stating that
the $(F_k)$ family of functions is ordered pointwise:
$F_k(n) \le F_{k+1}(n)$ for all $k\ge 1$ and $n\ge 0$.
Even if we will not compare further the functions of
this $(F_k)$ family, two other aspects of the previous article
will be relevant here.

First, we introduced there a family of
substitutions\footnote{A~substitution (or morphism) on an alphabet $A$ is a map $\tau: A^* \to A^*$ satisfying $\tau(uv) = \tau(u) \tau(v)$ for all $u,v \in A^*$, where $A^*$ denotes the set of finite words with letters in~$A$.
The map $\tau$ is therefore defined by its value on the letters of~$A$, and it is extended in a natural way to infinite words (or sequences) $w = w[0] w[1] \cdots \in A^\infty$ by setting $\tau(x) = \tau(x[0]) \tau(x[1]) \cdots$.
For more on substitutions, see for instance~\cite{Pytheas}.}
 $(\tau_k)$ and infinite morphic words $(x_k)$
and length functions $(L_k)$:

\begin{definition}
\label{defL}
For $k \ge 1$, let $\tau_k$ be the substitution
 on the alphabet $\{1,2,\dots,k\}$ defined by
\[
\begin{aligned}
\tau_k:\, k &\mapsto k1, \\
i & \mapsto i{+}1 \quad \mbox{for}\ 1 \le i < k.
\end{aligned}
\]
Let $x_k = x_k[0] x_k[1] \cdots \in \{1,2,\dots,k\}^\infty$ be the
fixed point of~$\tau_k$ and let $L_k$ be the following function:
\[
L_k:\, \mathbb{N} \to \mathbb{N}, \quad n \mapsto \big|\tau_k(x_k[0{:}n))\big|
\]
where $x_k[0{:}n)$ is the prefix of $x_k$ of size $n$, i.e.,
$x_k[0]\cdots x_k[n{-}1]$.
\end{definition}

These notions were shown to be deeply related with $F_k$,
in particular $L_k$ forms a Galois connection with $F_k$, i.e.,
almost a bijection, more on that in Section~\ref{s:basic} below.
The main result of monotonicity for the $(F_k)$ family was then
converted to a similar statement on $(L_k)$ and then proved by
investigating the letter configurations in the morphic words $(x_k)$.
In the present article, these lengths $L_k$ will receive new
descriptions, in particular as left shifts in some particular numeration
systems in Section~\ref{s:numrep}, while the words $(x_k)$ will
also be related with these numeration systems in Section~\ref{s:words}.

Secondly, we studied
in~\cite[Sect.~\ref*{article1:s:infini}]{LetouzeyShuoSteiner2024}
the infinitary behavior of $F_k$ and showed in particular that
$\lim \tfrac{1}{n}F_k(n)$
exists and is the unique positive zero of the polynomial $X^k{+}X{-}1$.
Noting $\alpha_k$ this zero, it amounts to say that $F_k(n)$ admits
$\alpha_k\,n$ as linear equivalent, a fact already known to
Dilcher~\cite{Dilcher1993}.
Here, we investigate in Section~\ref{s:distF}
the discrepancy of each $F_k$, i.e., the maximal distance
between $F_k$ and its linear equivalent. For that, the polynomial
$X^k{+}X{-}1$ and its reciprocal $X^k{-}X^{k-1}{-}1$ are studied
more carefully in Section~\ref{s:poly}, once again revisiting former
results by Dilcher~\cite{Dilcher1993}.

\subsection{The Coq artifact}

Almost all the proofs presented in this article have been formally certified
using the Coq/Rocq proof assistant~\cite{Coq}.
The files of this Coq
development are freely available~\cite{LetouzeyCoqDev},
the recommended entry point
to read alongside this article being
\begin{center}
\url{https://github.com/letouzey/hofstadter_g/blob/main/Article2.v}
\end{center}

This formal development ensures precise
definitions and statements and rules out any reasoning errors
or rounding errors during
the proofs. It can hence serve as a reference for the interested reader.
The current article tries to be faithful to this formal work
while staying readable by a large audience, at the cost of possible
remaining mistakes during the transcription.

Our Coq development can be machine-checked again by recent
installations of Coq, see the joint \texttt{README.md} file (the
authors used Coq version~8.16).
All the first part up to Section~\ref{s:words}
have been formalized within Coq core logic, without any extra axioms,
as may be checked via the command \texttt{Print Assumptions} on our
theorems. The parts corresponding to
Section~\ref{s:poly} and after involve real numbers and
hence rely on
some Coq standard libraries that declare four logical axioms,
in particular the axiom of excluded middle.

Some differences remain between the proofs presented in this article
and the corresponding Coq development:
\begin{itemize}
\item In Proposition~\ref{p:secondary}, we rely on the fact that
  the Plastic Number is the smallest Pisot number, an important result by
  Siegel~\cite{Siegel1944}. Lacking a full Coq
  proof of this result, we managed instead to
  formalize in Coq the first pages of the initial proof by Siegel
  and then adapted it to suit our need, namely that for $k\ge 6$,
  the polynomial $X^k-X^{k-1}-1$ does have secondary zeros
  of modulus strictly more that 1. More details can be found in our file
  {\tt SecondRoots.v}.

\item The Coq proofs corresponding to
Theorems~\ref{t:supdelta} and~\ref{t:supdelta5} only state that
the discrepancies $\Delta_k$ are infinite when $k\ge 5$, but do not
express yet the rate of divergence, which is
logarithmic for $k=5$ and a power function for $k\ge 6$.
These extra details were not formalized in Coq due to lack of time, but
no major difficulties is expected here should we try to do it later.
For the current Coq proofs, see {\tt Freq.delta\_sup\_qgen}
and {\tt LimCase4.delta\_sup\_q4}.

\item Finally, the Appendix~\ref{s:polycoeff} has not been considered
  yet in Coq.

\end{itemize}

\subsection{Summary of the current article}

Section~\ref{s:basic} recalls some earlier basic results about
functions $F_k$. Then Section~\ref{s:fibo} investigates a family
$\A{k}{p}$ of Fibonacci-like sequences and their relations with $F_k$.
This is extended in Section~\ref{s:numrep} where
numeration systems based on these $\A{k}{p}$ sequences are considered
by generalizing the Zeckendorf theorem, then $F_k$ is shown to be
a right shift on these digital
expansions (with a specific treatment of the lowest digit).
Section~\ref{s:words} translates these results in term of infinite
morphic words.
In Section~\ref{s:poly}, we study the polynomials corresponding to all these
recurrences, with a particular attention to the modulus of their zeros.
Section~\ref{s:Aalg} provides a detailed description of the
Fibonacci-like $\A{k}{p}$ numbers as linear combinations of powers of
these zeros.
Finally, Section~\ref{s:distF} assembles most of the previous pieces
for a study of the discrepancy of each $F_k$ i.e., the maximal distance between
$F_k$ and its linear equivalent. This discrepancy is
finite exactly when $k\le 4$, in particular we prove it to be less
than 1 for $k=3$ and less than 2 for $k=4$, before stating various
consequences and encountering a nice fractal when $k=3$.

\section{Basic properties of function \texorpdfstring{$F_k$}{\it F}}
\label{s:basic}

We recall now some earlier basic properties, in particular
Proposition~\ref*{article1:basicF} of~\cite{LetouzeyShuoSteiner2024},
where $\partial F_k^j(n)$ denotes the difference
$F_k^j(n{+}1)-F_k^j(n)$:

\begin{proposition}
\label{p:basicF}
For all $j,k \ge 1$, the function $F_k^j$ satisfies the following
basic properties:
\begin{enumerate}[label=(\alph*)]
\itemsep.5ex
\item $F_k^j(1) = F_k^j(2) = 1$,
\item $F_k^j(n) \ge 1$ whenever $n\ge 1$,
\item $F_k^j(n) < n$ whenever $n\ge 2$,
\item for all $n>0$, $\partial F_k(n) = 1 - \partial F_k^k(n{-}1)$,
\item hence $\partial F_k^j(n) \in \{0,1\}$ for all $n\ge 0$,
\item the function $F_k^j$ is monotonically increasing and onto (but not one-to-one).
\end{enumerate}
\end{proposition}

We also proved
earlier~\cite[Prop.~\ref*{article1:p:LCF}]{LetouzeyShuoSteiner2024}
that the function $L_k$ introduced in Definition~\ref{defL} can also be
expressed in terms of $F_k$:
\begin{proposition}
\label{p:L_as_F}
For $k\ge 1$ and $n\ge 0$, we have $L_k(n) = n+F_k^{k-1}(n)$.
\end{proposition}
Actually, this could have been an alternative definition for $L_k$,
skipping all references of morphic words until Section~\ref{s:words}.
One way or the other, we have the following important relation
between $F_k$ and $L_k$:

\begin{proposition} \label{p:FL}
For $k\ge 1$ and $n\ge 0$, we have $F_k(L_k(n))=n$ and
$L_k(F_k(n))\in\{n,n{+}1\}$.
\end{proposition}
\begin{proof}
See \cite[Sect.~\ref*{article1:s:relating_F_L}]{LetouzeyShuoSteiner2024}.
Alternatively, here is a more direct proof, just relying on
Propositions~\ref{p:basicF} and~\ref{p:L_as_F}:
\[
L_k(F_k(n)) = F_k(n) + F_k^k(n) = F_k(n) + (n{+}1-F_k(n{+}1)) =
n{+}1-\partial F_k(n).
\]
By Proposition~\ref{p:basicF}, $\partial F_k(n)\in\{0,1\}$ hence
$L_k(F_k(n))\in\{n,n{+}1\}$.

For the first part of the statement, let us pick $m$ such that
$F_k(m)=n$ ($F_k$ is onto, see Proposition~\ref{p:basicF}).
Proceeding as before, but for $m$ this time:
\[
F_k(L_k(n)) = F_k(L_k(F_k(m))) = F_k(m{+}1-\partial F_k(m))
\]
Now, either $\partial F_k(m) = 0$ and hence $F_k(L_k(n)) = F_k(m{+}1) =
F_k(m) = n$. Or $\partial F_k(m) = 1$ and hence
$F_k(L_k(n)) = F_k(m{+}1{-}1) = F_k(m) = n$.
\end{proof}

\begin{corollary}
For all $k\ge 1$, the functions $F_k$ and $L_k$
form a \emph{Galois connection} between $\mathbb{N}$ and itself
(with $F_k$ as \emph{left adjoint} and $L_k$ as \emph{right adjoint}).
Indeed, for all $n,m \ge 0$ we have $F_k(n) \le m$ iff $n \le
L_k(m)$. Moreover, this Galois connection is said to be a
\emph{Galois insertion} since $F_k\circ L_k = id$.
\end{corollary}
\begin{proof}
See~\cite[Cor.~\ref*{article1:galois}]{LetouzeyShuoSteiner2024}.
Or alternatively, by Proposition~\ref{p:basicF}, the function $F_k^{k-1}$ is
increasing, and hence $L_k(n)=n+F_k^{k-1}(n)$ is increasing as well.
If $F_k(n) \le m$, then $L_k(F_k(n)) \le L_k(m)$. By
Proposition~\ref{p:FL}, $L_k(F_k(n))\in\{n,n{+}1\}$ hence
$n\le L_k(F_k(n)\le L_k(m)$.
Conversely, $n\le L_k(m)$ implies $F_k(n) \le F_k(L_k(m)) = m$.
\end{proof}

In particular, $L_k(n)$ is the rightmost antecedent of $n$ by $F_k$.
Note that there cannot exist more than two antecedents of a given
point by $F_k$. Indeed, if $F_k(a)=F_k(a{+}1)$ then $F_k^k(a) = F_k^k(a{+}1)$
hence $\partial F_k^k(a)=0$ hence $\partial F_k(a{+}1)=1$
by Proposition~\ref{p:basicF}, so $F_k(a{+}2) > F_k(a)$.
Similarly, if $F_k(a)=F_k(a{+}1)$ then $a\neq 0$
and $F_k(a{-}1)<F_k(a)$, otherwise $\partial F_k(a{-}1)=0$ hence
$\partial F_k^k(a{-}1)=0$ and $\partial F_k(a)=1$.

\section{Related Fibonacci-like recurrence}
\label{s:fibo}

A family of Fibonacci-like recursive sequences is of
particular interest here. It it obtained by
generalizing the Fibonacci recurrence in the following way:
instead of adding the two previous terms of the sequence, we add the
previous term and the $k$-th earlier term, for some $k\ge 1$.
We will see below that these sequences could also be expressed as
$L_k^p(1)$, and that $F_k$ shifts down the corresponding
sequence by one step.

\begin{definition}
For any integer $k\ge 1$, let
$(\A{k}{p})_{p\in\mathbb{N}}$ be the following linear recurrence:
\[
\A{k}{p} =
\begin{cases}
 p{+}1, & \text{if $0 \le p < k$;} \\
 \A{k}{p-1} + \A{k}{p-k}, & \text{if $p \ge k$.}
\end{cases}
\]
\end{definition}

\begin{figure}[ht]
\begin{tabular}{c|rrrrrrrrrrr}
p & 0 & 1 & 2 & 3 & 4 & 5 & 6 & 7 & 8 & 9 & 10 \\
\hline
\A{1}{p} & 1 & 2 & 4 & 8 & 16 & 32 & 64 & 128 & 256 & 512 & 1024 \\
\A{2}{p} & 1 & 2 & 3 & 5 & 8 & 13 & 21 & 34 & 55 & 89 & 144 \\
\A{3}{p} & 1 & 2 & 3 & 4 & 6 & 9 & 13 & 19 & 28 & 41 & 60 \\
\A{4}{p} & 1 & 2 & 3 & 4 & 5 & 7 & 10 & 14 & 19 & 26 & 36 \\
\A{4}{p} & 1 & 2 & 3 & 4 & 5 & 6 & 8  & 11 & 15 & 20 & 26
\end{tabular}
\caption{Initial values of $\A{k}{p}$.}
\end{figure}

Hence $\A{1}{p} = 2^p$ for all $p\ge 0$ and $(\A{2}{p})$ is the Fibonacci
sequence with initial values
$1$ $2$ $3$ $5$ $8$ (no initial 0, only one 1).
The sequence $(\A{3}{p})$ is also known as Narayana's Cows, see
OEIS~\seqnum{A930} but with an index shift. See also the recent
study by Shallit~\cite{Shallit2025}.
Moreover, this whole family of recursive sequences has already been
considered many
times~\cite{MeekVanRees84,Dilcher1993,Kimberling95,EriksenAnderson12}.
Several of these earlier papers also use the $\A{k}{p}$ numbers
as basis for alternative digital expansions, as done
in the following Section~\ref{s:numrep}.

Note that these $\A{k}{p}$ numbers should not be confused with the
number of ordered $k$-arrangements of $n$, which is denoted
$A_n^k$ in French texts.
Interestingly, the $\A{k}{p}$ numbers also have a nice combinatorial
interpretation:

\begin{proposition}\label{p:Acombi}
 For $k\ge 1$ and $p\ge 0$, $\A{k}{p}$ is the
  number of subsets of $\{1,\ldots,p\}$ with a distance of at least $k$
 between elements.
\end{proposition}
\begin{proof}
Let $k\ge 1$. We proceed by strong induction over $p$.
For $p < k$, the distance constraint ensures that the considered
subsets are either empty or singleton, so we have indeed $p{+}1$ such
subsets. And for $p\ge k$, we partition into the subsets
that contain $p$ and the ones that do not. By induction hypothesis,
the first group
has cardinal $\A{k}{p-k}$ (since the second-largest element
of these subsets is at most $p{-}k$ in this case),
while the second group has cardinal
$\A{k}{p-1}$, hence a total cardinal of $\A{k}{p}$.
\end{proof}

Obviously, all $\A{k}{p}$ are strictly positive, hence for a given
$k$, $(\A{k}{p})$ is strictly increasing over $p$.
Now, as an easy consequence of the previous combinatorial
interpretation, for a given $p\ge 1$
the sequence $(\A{k}{p})_{k>0}$ is decreasing from $\A{1}{p}=2^p$
till $\A{p}{p}=p{+}1$ and stationary afterwards. Indeed, the
rule $\A{k}{p} = p{+}1$ for initial terms actually holds even when
$p=k$, since $\A{p}{p} = \A{p}{p-1} + \A{p}{p-p} = p{+}1$.

Interestingly, we can also extend the recurrence rule and
make it handle all terms except $\A{k}{0}$:
\begin{equation}
\label{Agenrec}
\A{k}{p} = \A{k}{p-1} + \A{k}{p\ominus k} \qquad\quad \text{for all}\ p\ge 1
\end{equation}
where $a\ominus b = \max(0,a-b)$ is the total subtraction from
$\mathbb{N}^2$ to $\mathbb{N}$ that rounds negative results to 0.

As promised, the functions $F_k$ and $L_k$ have a specific
behavior on these $\A{k}{p}$ numbers.
\begin{proposition}
\label{p:FA}
For all $k \ge 1$ and $p,j \ge 0$,
\begin{align*}
F_k^j(\A{k}{p}) &= \A{k}{p\ominus j} \\
L_k^j(\A{k}{p}) &= \A{k}{p+j}.
\end{align*}
Hence $\A{k}{p}$ can also be expressed as $L_k^p(1)$.
\end{proposition}
\begin{proof}
We prove the first equation by induction over $p$.
The case $p=0$ is obvious since
$F_k^j(1)=1$. Now let $p>0$ and assume the equation is true for
$p-1$ and for all $j\ge 0$. Of course $F_k^0(\A{k}{p}) =
\A{k}{p+0}$. Now, for $j\ge 0$, note that here
$p{\ominus} (j{+}1) = p{-}1\ominus j$ and
$p{\ominus} k = (p{-}1)\ominus(k{-}1)$, hence
\begin{align*}
F_k^{j+1}(\A{k}{p}) &= F_k^j(F_k(\A{k}{p-1}+\A{k}{p\ominus k}))
                      \qquad \mbox{(by Eq.~\eqref{Agenrec})}\\
                    &= F_k^j(F_k(\A{k}{p-1}+F_k^{k-1}(\A{k}{p-1}))
                       \qquad \mbox{(by I.H.)}\\
                    &= F_k^j(F_k(L_k(\A{k}{p-1})) \\
                    &= F_k^j(\A{k}{p-1})
                       \qquad \mbox{(by Prop.~\ref{p:FL})} \\
                    &= \A{k}{p-1\ominus j}
                       \qquad \mbox{(by I.H.)}\\
                    &= \A{k}{p\ominus(j+1)}
\end{align*}
Now, from the definition of $L_k$ and Equation~\eqref{Agenrec}:
\[
L_k(\A{k}{p}) = \A{k}{p}+F_k^{k-1}(\A{k}{p})
              = \A{k}{p}+\A{k}{p\ominus(k-1)}
              = \A{k}{p}+\A{k}{(p+1)\ominus k}
              = \A{k}{p+1}.
\]
This extends to $L_k^j$ via a direct induction on $j$. Finally,
$L_k^p(1) = L_k^p(\A{k}{0}) = \A{k}{p}$.
\end{proof}

This last proposition will be generalized in the next Section,
where Theorem~\ref{t:Fshift} and Proposition~\ref{p:Lshift}
establish the behavior of $F_k^j$ and $L_k^j$ on some sums
of $\A{k}{p}$ numbers.

\section{Fibonacci-like digital expansions}
\label{s:numrep}

For a given $k\ge 1$, the Fibonacci-like numbers $(\A{k}{p})$
seen in section~\ref{s:fibo} can be used as a base for a family of digital
expansions.
In these numeration systems, we will see that $L_k$ and $F_k$ behave as
bitwise operations, namely left and right shifts (with a particular
treatment of the lowest digit in the case of $F_k$). This extends
Proposition~\ref{p:FA}.

Such a digital expansion can be described in a standard way
through sequences of digits $(d_i)_{i\in\mathbb{N}}$ that
are ultimately null, leading to sums $\Sigma_{i=0}^\infty\ d_i\A{k}{i}$.
As usual for such a digital expansion, we may display a number
via its digits, with the least-significant digit on the right.
Nonetheless, we favor here an alternative presentation based
on the positions of the non-null digits, calling this
a $k$-\emph{decomposition}. In the rare situations where some digits
will be 2 or more, we repeat the corresponding positions.
This approach alleviates the technical lemmas to come.

\begin{definition}
For $k\ge 1$, a $k$-decomposition is a finite sequence
$D=p_0 \cdots p_m$ of
increasing natural numbers, with possible repetitions.
Such a sequence may be empty, noted $D=\varnothing$.
The sum of a $k$-decomposition $D=p_0 \cdots p_m$ is
$\Sigma_k(D) = \A{k}{p_0}+\cdots+\A{k}{p_m}$.
For $n\ge 0$,
a $k$-decomposition of $n$ is any $k$-decomposition $D$ such that
 $\Sigma_k(D)=n$.
A $k$-decomposition $D$ will be said \emph{canonical}
(resp., \emph{lax}) when the distance between the elements of $D$ is
at least $k$ (resp., at least $k{-}1$).
\end{definition}

Note that for $k\ge 1$, the $k$-decompositions that are canonical
cannot contain
repeated positions, and can hence be viewed as simple sets of
natural numbers (with sufficient distance between them).
For simplicity sake, we (ab)use set notation for these
$k$-decompositions, when this remains clear enough, for instance
$D\cup\{p\}$ or $D\smallsetminus\{p\}$ for adding position $p$ to a
decomposition $D$ (or removing it).
Nonetheless, the exact description remains here through
finite sequences, allowing the rare but complex
situations where repeated positions are needed,
see in particular the Definition~\ref{def_rshift} below.

As an example, let us try $k=2$, i.e., decompose in sums of Fibonacci numbers:
for instance the number 17 is equal to
$13 + 3 + 1 = \A{2}{5} + \A{2}{2} + \A{2}{0}$
so 17 have a digit representation of $100101$, corresponding to a
(canonical) $2$-decomposition of $0,2,5$.
This decomposition is indeed canonical since the positions are
separated by at least 2. Other decompositions are possible,
for instance $17 = 8 + 5 + 3 + 1 = 11101$, i.e.,
the lax $2$-decomposition $0,2,3,4$.
Considering higher digits allows more decompositions, for example
$17 = 2*8 + 1 = 20001$, i.e., the $2$-decomposition $0,4,4$.
This one is not canonical nor lax, and will hence be avoided as
much as possible.
Now, if we switch to $k=3$, the same example 17 is now
$13 + 4 = \A{3}{6} + \A{3}{3} = 1001000$, hence the
canonical $3$-decomposition $3,6$.
Note that for $k=1$, a canonical $1$-decomposition is actually
a sum of distinct powers of two, i.e., the usual base-2
representation, for instance $17 = 16 + 1 = 10001$.
In this case $k=1$, considering lax decompositions already
implies having digits of 2 or more, i.e., repeated positions.

The least digit has always the same weight
independently of $k$, since
$\A{k}{0} = 1$. Similarly, the second digit always weights $\A{k}{1} = 2$.
All other digits have weights that depend on $k$, for instance
$\A{k}{2}\in\{3,4\}$ and $\A{k}{3}\in\{4,5,8\}$. For a given $k$,
the corresponding numeration system is clearly redundant, but
with unique canonical decomposition.
Indeed, Zeckendorf's theorem can be adapted for this generalized setting,
as already seen for instance by Kimberling~\cite{Kimberling95} or
Eriksen and Anderson~\cite{EriksenAnderson12}:

\begin{theorem}[Zeckendorf]
\label{t:zeck}
For $k\ge 1$ and $n\ge 0$, there exists a unique canonical
$k$-decomposition of $n$. We note it $D_k(n)$.
\end{theorem}
\begin{proof}
The proof of this theorem is quite standard. First we build by strong
induction a canonical $k$-decomposition of $n$. For $n=0$ we take
the empty decomposition, whereas for
$n>0$ we consider first the highest position $p$ such that $\A{k}{p} \le n$,
then build recursively a canonical $k$-decomposition $D$ of the rest
$n-\A{k}{p}$. Now
$D \cup \{p\}$ is a satisfactory decomposition of $n$.
Indeed, the choice criterion of $p$ implies $\A{k}{p} \le n < \A{k}{p+1}$.
If $i<k$, then $\A{k}{p+1} = 1 + \A{k}{p}$ and hence $n$ is exactly $\A{k}{p}$
and the decomposition stops there, with a singleton decomposition
which is obviously correct and canonical. Otherwise $i\ge k$ and
so $\A{k}{p+1} = \A{k}{p} + \A{k}{p+1-k}$ hence $n-\A{k}{p} < \A{k}{p+1-k}$.
So the largest index in the decomposition $D$ cannot be $p{+}1{-}k$ or
more, so $D \cup\{p\}$ is still canonical. Repeating this will hence
build a decomposition that is indeed correct and canonical.

For proving the uniqueness, the key ingredient is the following:
if a non-empty canonical decomposition $D$ of $n$ has $p$
for highest position, then $n < \A{k}{p+1}$. This is provable by induction
on the cardinal of $D$, and then considering the sub-decomposition
$D{\smallsetminus}\{p\}$. With this key fact, one can easily prove that
two canonical
decompositions of the same number must have the same highest positions,
and then recursively the same other elements.
\end{proof}

In general, a number may admit several lax $k$-decompositions
beside the canonical one. Let us call $\norm_k(D)$ the unique
canonical $k$-decomposition having the same sum as $D$. Obviously,
it can be computed by sum and decompose: $\norm_k(D) = D_k(\Sigma_k(D))$.
When $D$ is
lax, a more efficient algorithm is to ``repair'' this almost canonical
decomposition: we repeatedly
find the highest pair of indices in $D$ at distance $k{-}1$, say
$p{-}1$ and $p{-}k$, and replace them by $p$. On $\Sigma_k(D)$, this
amounts to change terms $\A{k}{p-1} + \A{k}{p-k}$ with $\A{k}{p}$.
The choice of the highest possible pair implies that the new decomposition
is at least lax again, while having a strictly lower cardinal that $D$.
Repeating the process will hence stop on
a canonical decomposition of equal sum, in no more steps than the cardinal
of~$D$.

Let us express now the decomposition of $n+1$ in terms of the
decomposition of $n$. For that, we need to discriminate on the least
element of this decomposition.

\begin{definition}
The rank of a decomposition $D$ is the least element of $D$,
or $+\infty$ when $D=\varnothing$.
The $k$-rank of a number $n$ is the rank of $D_k(n)$.
\end{definition}

\begin{proposition}
\label{DkSucc}
For $k\ge 1$ and $n\ge 0$, let us call $r$ the $k$-rank of $n$.
The canonical $k$-decomposition of $n{+}1$ is
\[
D_k(n{+}1) =
\begin{cases}
 D_k(n) \cup \{0\}, & \text{if $r \ge k$;} \\
 \norm_k(D_k(n)\smallsetminus \{r\}\cup\{r{+}1\}), & \text{otherwise.}
\end{cases}
\]
\end{proposition}

We now generalize some bitwise operations of base 2.

\begin{definition}
\label{def_rshift}
For a $k$-decomposition $D = p_0 \cdots p_m$ and $q\ge 0$,
the left shift of $D$ by $q$ is adding $q$ to all positions, i.e.,
\[
 D\ll q = (p_0{+}q) \cdots (p_m{+}q).
\]
The upper right shift of $D$ by $q$ is
\[
 D\gg_+ q = (p_0{\ominus}q) \cdots (p_m{\ominus}q).
\]
Finally, the (usual) right shift of $D$ by $q$ is obtained by
considering only the positions $p_\ell \cdots p_m$ in D that are
greater or equal to $q$ (if any), and then:
\[
 D\gg q = (p_\ell{-}q) \cdots (p_m{-}q).
\]
\end{definition}

Recall that $p\ominus q = \max(0,p-q)$, so an upper right shift
may produce repeated position~0, even when starting from a canonical
decomposition. At least $D \gg_+ 1$ is known to be lax when
$D$ is a canonical $k$-decomposition with $k\ge 1$, while
$D \ll q$ and $D \gg q$ are still canonical.
Still for a canonical $k$-decomposition $D$, $D \gg_+ 1$ is
either equal to $D \gg 1$ or differ by an extra initial 0,
depending on whether $D$ contains 0 or not.
In the particular case $k=1$, $D_k(n) \gg 1$ and $D_k(n) \gg_+ 1$
correspond respectively to $\lfloor n/2 \rfloor$ and
$\lceil n/2 \rceil$.

Meek and Van Rees~\cite{MeekVanRees84} showed that the right shift
$\gg 1$ on $k$-decompositions is actually a recursive function
very similar to $F_k$, with just an extra $-1$ in the recursive definition.
\begin{definition}
For $k\ge 1$, let $\tF_k$ be defined recursively by
\[
\tF_k:\, \mathbb{N} \to \mathbb{N}, \quad n \mapsto
\begin{cases}
 0, & \text{if $n = 0$;} \\
 n-1-\tF_k^k(n{-}1), & \text{otherwise.}
\end{cases}
\]
\end{definition}
As noticed by Meek and Van Rees, $\tF_k$ is a translated version of $F_k$.
\begin{proposition}
For all $k\ge 1$ and $n\ge 0$, $\tF_k(n)=F_k(n+1)-1$.
\end{proposition}
Here comes the result of Meek and Van Rees in our notation.
\begin{theorem}[Meek and Van Rees]
For all $k\ge 1$ and $n\ge 0$, $D_k(\tF_k(n)) = D_k(n) \gg 1$ hence
$\tF_k(n) = \Sigma_k (D_k(n) \gg 1)$.
\end{theorem}
Combining these two last results, it is possible to prove that $F_k$ is
actually the upper right shift $\gg_+ 1$. Instead, we propose here
a direct proof of this fact, without translation to $\tF_k$.
The downside of this direct proof is that lax decompositions have to
be considered, hence the following technical lemma.

\begin{proposition}
\label{shuntnorm}
Let $k\ge 1$. For a lax $k$-decomposition $D$ and a number $q < k$, we have
\[
\Sigma_k(\norm_k(D) \gg_+ q) = \Sigma_k(D \gg_+ q)
\]
\end{proposition}
\begin{proof}
The normalization repeatedly replaces pairs of positions such as
$p{-}1,p{-}k$ into $p$, and exploits the
equation $\A{k}{p-1}+\A{k}{p-k} = \A{k}{p}$ for keeping the
overall sum unchanged. Here, $p{-}k$ is a legal position, hence
$p\ge k > q$ and so $p\ominus q = p-q \ge 1$.
Now, if we shift by $\gg_+ q$ before summing,
the corresponding equation
$\A{k}{p-1\ominus q}+\A{k}{p-k\ominus q} = \A{k}{p\ominus q}$
still hold. Indeed,
Equation \eqref{Agenrec} for $p\ominus q$ gives
$\A{k}{p\ominus q} = \A{k}{p\ominus q-1}+\A{k}{p\ominus q\ominus k}$
while here
$p\ominus q\ominus k = p\ominus k\ominus q$ and
$p\ominus q-1 = p-1\ominus q$.
\end{proof}

\begin{theorem}
\label{t:Fshift}
For all $k\ge 1$ and $n\ge 0$,
\[
F_k(n) = \Sigma_k(D_k(n)\gg_+ 1).
\]
More generally, for all $1\le q \le k$, $F_k^q(n) = \Sigma_k(D_k(n)\gg_+ q)$.
\end{theorem}
\begin{proof}
The case $k=1$ is a consequence of the identity
$F_1(n) = \lceil n/2 \rceil$
already mentioned earlier. We now assume $k\ge 2$.
Let us note $G_k(n) = \Sigma_k(D_k(n)\gg_+ 1)$.
We first prove that
\begin{equation}
\label{Gkp}
G_k^q(n) = \Sigma_k(D_k(n)\gg_+ q)
\end{equation}
for all $1 \le q \le k$ and $n\ge 0$. The case $q=1$ is straightforward.
Let $1 \le q < k$ and assume \eqref{Gkp} for this $q$ and all $n$.
By noting $D = D_k(n)\gg_+ 1$ we have
\[
G_k^{q+1}(n) = G_k^q (G_k(n)) = \Sigma_k(D_k(\Sigma_k(D)) \gg_+ q)
\]
Due to its definition, $D$ is at least a lax decomposition, hence
$D_k(\Sigma_k(D)) = \norm_k(D)$ and by Proposition \ref{shuntnorm}
we obtain
\[
G_k^{q+1}(n) = \Sigma_k(D \gg_+ q) = \Sigma_k (D_k(n) \gg_+ 1 \gg_+ q)
= \Sigma_k(D_k(n) \gg_+ (q{+}1)).
\]
For the last equality, note the law
$a\ominus b\ominus c = a\ominus (b+c)$, hence
two upper right shifts may be grouped into a single one.
By induction on $q$, this concludes the proof of \eqref{Gkp}.

Now for proving $F_k=G_k$ we show that $G_k(0)=0$ (which is obvious)
and that $G_k$ satisfies the same recursive equation as $F_k$.
The equation about $F_k^q$ then follows from \eqref{Gkp}.
So let $n\ge 0$ and consider $G_k(n{+}1)+G_k^k(n)$.
If $n=0$ this sum is clearly $1=n{+}1$. Assume now $n>0$ and
note $r = \rank_k(n)$ and $D = D_k(n)\smallsetminus\{r\}$. First, if
$r < k$, then $D_k(n{+}1) = \norm_k(D\cup\{r{+}1\})$, so
\begin{align*}
G_k(n{+}1)+G_k^k(n) &=
  \Sigma_k(\norm_k(\{r{+}1\}\cup D) \gg_+ 1) + \Sigma_k((\{r\}\cup D) \gg_+ k)\\
&=\Sigma_k((\{r{+}1\}\cup D) \gg_+ 1) + \Sigma_k((\{r\}\cup D) \gg_+ k)
  \ \ \ \text{(by Prop. \ref{shuntnorm})} \\
&=\A{k}{r{+}1\ominus 1} + \A{k}{r\ominus k} + \Sigma_k(D \gg_+ 1) +
 \Sigma_k(D \gg_+ k)\\
&=\A{k}{r} + \A{k}{0} + \Sigma_k(D)\\
&=1 + \Sigma_k(\{r\}\cup D) \\
&=1 + n
\end{align*}
We exploited here the fact that all elements in $D$ are at least $k{+}r$
hence non-zero. So $\Sigma_k(D \gg_+ 1) + \Sigma_k(D \gg_+ k)$
is $\Sigma_k(D)$ via repeated use of Equation~\ref{Agenrec}.

Otherwise rank $r \ge k$ and $D_k(n{+}1) = \{0\}\cup D_k(n)$ and
\begin{align*}
G_k(n{+}1)+G_k^k(n) &=
 \Sigma_k((\{0\}\cup D_k(n)) \gg_+ 1) + \Sigma_k(D_k(n) \gg_+ k) \\
 &= \A{k}{0\ominus 1} + \Sigma_k(D_k(n) \gg_+ 1) + \Sigma_k(D_k(n)
 \gg_+ k) \\
 &= 1 + \Sigma_k(D_k(n)) \\
 &= 1 + n
\end{align*}
This time, it is $D_k(n)$ whose elements are non-zero, and we
handled the sum of its shifts just as before.
\end{proof}

As noticed earlier, $D_k(n) \gg_+ 1$ is at least a lax
$k$-decomposition, hence a renomalization may be necessary
when computing the canonical $k$-decomposition of $F_k(n)$:
\[
D_k(F_k(n)) = \norm_k(D_k(n) \gg_+ 1).
\]
As an alternative formulation,
$F_k(n)$ is also the (usual) right shift of the
$k$-decomposition of $n$ (i.e., $\tF_k(n)$),
plus the least digit of this decomposition of $n$
(i.e., its ``$k$-parity''). When this digit is $1$, this may trigger
a renormalization of the decomposition.

We now relate the $k$-rank and the differences $\partial F_k^q$.
From Section~\ref{s:basic}, we recall that $\partial F_k^q(n)$
denotes $F_k^q(n{+}1)-F_k^q(n)$ and that it is always either 0 or 1.

\begin{theorem}
\label{t:flatrank}
For all $k\ge 1$ and $n\ge 0$ and $0\le q \le k$,
we have $\partial F_k^q(n) = 0$ iff $\rank_k(n)< q$.
In particular $\partial F_k(n) = 0$ iff $\rank_k(n) = 0$.
\end{theorem}
\begin{proof}
First, for $q=0$, $\partial F_k^0(n)=1\neq 0$ and $\rank_k(n) \nless 0$.
Moreover, for $q=k=1$, $F_1(n)=\lceil n/2 \rceil$ and
$\partial F_1(n) = 0$ when $n$ is odd which is indeed when $\rank_1(n)= 0$.

Consider now $1 \le q < k$, leaving away the case $q=k$ for the moment
to be able to benefit from Proposition~\ref{shuntnorm}.

If $n=0$, $F_k^q(1)=1\neq 0 = F_k^q(0)$ and $\rank_k(0) = +\infty > q$.

If $n>0$ and $\rank_k(n)<k$,
we reuse the notation of the last proof: $r = \rank_k(n)$ and
$D=D_k(n)\smallsetminus\{r\}$. The last theorem states that
\begin{align*}
\partial F_k^q(n)
&= \Sigma_k(D_k(n{+}1) \gg_+ q) - \Sigma_k(D_k(n) \gg_+ q) \\
&= \Sigma_k(\norm_k(\{r{+}1\}\cup D) \gg_+ q) - \Sigma_k((\{r\}\cup D)
\gg_+ q) \\
&= \Sigma_k((\{r{+}1\}\cup D) \gg_+ q) - \Sigma_k((\{r\}\cup D) \gg_+ q)
                \ \ \ \text{(by Prop. \ref{shuntnorm})} \\
&= \A{k}{r{+}1\ominus q} - \A{k}{r \ominus q}
\end{align*}
This is indeed 1 when $q\le r<k$, since $r{+}1\ominus q = 1+(r\ominus q) < k$.
Otherwise $r<q$ and $r{+}1\ominus q = 0 = r\ominus q$ hence
$\partial F_k^q(n) = 0$.

If $n>0$ and $\rank_k(n)\ge k$ and in particular $\rank_k(n)\ge q$:
\begin{align*}
\partial F_k^q(n)
&= \Sigma_k(D_k(n{+}1) \gg_+ q) - \Sigma_k(D_k(n) \gg_+ q) \\
&= \Sigma_k((\{0\}\cup D_k(n)) \gg_+ q) - \Sigma_k(D_k(n) \gg_+ q) \\
&= \A{k}{0} \\
&= 1
\end{align*}

Finally, for the case $q=k$, we use the identity
$\partial F_k^k(n) = 1 - \partial F_k(n{+}1)$. We have just shown that
$\partial F_k(n{+}1) = 0$ iff $\rank_k(n{+}1)<1$. From
Proposition~\ref{DkSucc}, we can see that $\rank_k(n{+}1)=0$ iff
$\rank_k(n)\ge k$. Hence $\partial F_k^k(n) = 0$ iff $\rank_k(n)<k$.
\end{proof}

\begin{proposition}
\label{p:Lshift}
For all $k\ge 1$ and $n, q\ge 0$, we have
$L_k^q(n) = \Sigma_k(D_k(n) \ll q)$ and hence
$D_k(L_k^q(n)) = D_k(n) \ll q$.
\end{proposition}
\begin{proof}
First, for $q=1$, thanks to Proposition~\ref{p:L_as_F} and
Theorem~\ref{t:Fshift} and Equation~\eqref{Agenrec}:
\begin{align*}
L_k(n) &= n + F_k^{k-1}(n) \\
       &= \Sigma_k(D_k(n)) + \Sigma_k(D_k(n) \gg_+ (k-1)) \\
       &= \sum_{p\in D_k(n)} (\A{k}{p} + \A{k}{p\ominus(k-1)}) \\
       &= \sum_{p\in D_k(n)} \A{k}{p+1} \\
       &= \Sigma_k(D_k(n) \ll 1)
\end{align*}
Since $D_k(n) \ll 1$ is still a canonical $k$-decomposition, whose
sum is equal to $L_k(n)$, we can conclude $D_k(L_k(n)) = D_k(n) \ll 1$
by Theorem~\ref{t:zeck}.

This extends easily by induction on $q$. It is clear for $q=0,1$.
Now, assume $L^q(n) = \Sigma_k(D_k(n) \ll q)$ for a specific
$q$ and all $n$. Then:
\[
L_k^{q+1}(n) = L_k^q(L_k(n)) = \Sigma_k(D_k(L_k(n)) \ll q) =
\Sigma_k (D_k(n) \ll 1 \ll q)
\]
and we can conclude by regrouping the two successive left shifts.
\end{proof}

\section{Decompositions and morphic words}
\label{s:words}

We already presented in Definition~\ref{defL} the substitution
$\tau_k$ and the infinite morphic word $x_k$ that are intimately
related with each function $F_k$. In this section we show that
these words $x_k$ can also be expressed in term of $k$-decomposition
and $k$-rank. Before that, we recall an important link
between functions $F_k$ and words $x_k$,
see~\cite[Prop.~\ref*{article1:diffFx}]{LetouzeyShuoSteiner2024}.
Note that in the word $x_k$, positions are indexed from 0.

\begin{proposition}
\label{diffFx}
Consider $1\le j < k$ and $n\ge 0$. We have $x_k[n]=j$ if and only if
both $\partial F_k^{j-1}(n) = 1$ and $\partial F_k^j (n) = 0$.
Moreover for $k\ge 1$ we have $x_k[n]=k$ iff $\partial F_k^{k-1}(n) = 1$
(in this case $\partial F_k^{k}(n)$ could be either 0 or~1).
\end{proposition}

From that, the letters of the word $x_k$ can be expressed in terms of $k$-rank.

\begin{theorem}
For $k\ge 1$ and $n\ge 0$ we have $x_k[n] = \min(k,1{+}\rank_k(n))$.
\end{theorem}
\begin{proof}
Combination of Proposition~\ref{diffFx} and Theorem~\ref{t:flatrank}.
\end{proof}

We also recall that the substitution $\tau_k$ satisfies a recursive
equation extending the recursive rule of $\A{k}{p}$ numbers,
see~\cite[Prop.~\ref*{article1:p:taujk_rec}]{LetouzeyShuoSteiner2024}.

\begin{proposition}
\label{p:tau-alt}
 For all $k\ge 1$ and $j\ge 0$,
either $j\le k$ and $\tau_k^j(k)=k1\cdots j$, or $j\ge k$ and
$\tau_k^j(k) = \tau_k^{j-1}(k)\tau_k^{j-k}(k)$. Hence the following
equalities on words and lengths:
\begin{itemize}
\item $x_k[0{:}\A{k}{j}) = \tau_k^j(k)$
\item $\A{k}{j} = |\tau_k^j(k)| = L_k^j(1)$.
\end{itemize}
\end{proposition}
\begin{proof} Routine induction.
\end{proof}

The Theorem~\ref{t:zeck} can be extended to express decompositions
of prefixes of words~$x_k$.
\begin{definition}
For a $k$-decomposition $D = p_0 \cdots p_m$,
we define the corresponding word $W_k(D)$ as
\[
    W_k(D) = \tau_k^{p_m}(k)\cdots\tau_k^{p_0}(k).
\]
\end{definition}

Note that $|W_k(D)| = \Sigma_k(D)$. We recall from our conventions
that the positions $p_0 \cdots p_m$ are increasing:
$p_0 \le \cdots \le p_m$. Even more, they are strictly increasing
(by steps of at least $k$) for canonical $k$-decompositions.

\begin{theorem}
\label{t:WkDk}
For $k\ge 1$ and $n\ge 0$, the prefix of $x_k$ of size $n$ is
$x_k[0{:}n) = W_k(D_k(n))$.
Moreover, for $q\ge 0$, we have $\tau_k^q(x_k[0{:}n)) = W_k(D_k(n) \ll q)$.
\end{theorem}
\begin{proof}
We follow the existence part of Theorem~\ref{t:zeck}.
Once again we proceed by strong induction over $n$. If $n=0$, then
indeed $x_k[0{:}0) = \varepsilon = W_k(\varnothing) = W_k(D_k(0))$.
Otherwise for $n>0$ let $p$ be the highest position such that
$\A{k}{p} \le n < \A{k}{p+1}$.
If $p<k$, then $\A{k}{p+1} = 1 + \A{k}{p}$ and hence $n=\A{k}{p}$ and
$D_k(n)=\{p\}$. So by Proposition~\ref{p:tau-alt},
\[
x_k[0{:}n) = \tau_k^p(k) = W_k(\{p\}) = W_k(D_k(n)).
\]
Otherwise $p\ge k$.
Since $\A{k}{p} \le n < \A{k}{p+1}$, Proposition~\ref{p:tau-alt}
implies that $\tau_k^p(k)$ is a prefix of $x_k[0{:}n)$
which is a prefix of
$\tau_k^{p+1}(k) = \tau_k^{p}(k)\tau_k^{p-k}(k)$.
That means that $x_k[0{:}n)$ can be written $\tau_k^{p}(k)w$ where
$w$ is a prefix of $\tau_k^{p-k}(k)$ and hence a prefix of $x_k$.
Hence $w = x_k[0{:}m)$ if we call $m$ the desired length of $|w|$,
i.e., $m = n-\A{k}{p}$. But here $D_k(n) = \{p\}\cup D_k(m)$ and $m<n$
and $p$ is strictly above all positions in $D_k(m)$.
By induction hypothesis, $x_k[0{:}m) = W_k(D_k(m))$ and finally
\[
 x_k[0{:}n) = \tau_k^p(k)W_k(D_k(m)) = W_k(\{p\}\cup D_k(m)) =
 W_k(D_k(n)).
\]

Now, for any $k$-decomposition $D$, we have $\tau_k^q(W_k(D)) = W_k(D \ll q)$
by distributing the substitutions $\tau_k^q$ into the components of
$W_k(D)$. Hence
\[
 \tau_k^q(x_k[0{:}n)) = \tau_k^q(W_k(D_k(n))) = W_k(D_k(n) \ll q).
\]
\end{proof}

\section{Related polynomials and algebraic integers}
\label{s:poly}

For a finer description of the Fibonacci-like $(\A{k}{p})$ numbers in
Section~\ref{s:Aalg}, we investigate now some
corresponding polynomials. This is an extension
of our previous
work~\cite[Sect.~\ref*{article1:s:poly}]{LetouzeyShuoSteiner2024},
where we already presented the following polynomials and their
positive zeros. Most of this section was already in Dilcher~\cite{Dilcher1993}.

\begin{definition} \label{d:alphabeta}
For $k\ge 1$, we name $P_k(X) = X^k+X-1$ and
$Q_k(X) = X^k-X^{k-1}-1$. We name $\alpha_k$ (resp., $\beta_k$) the
unique positive zero of $P_k$ (resp., $Q_k$).
Note that $\beta_k = 1/\alpha_k$ since $P_k$ and $Q_k$ are reciprocal
polynomials.
\end{definition}

\begin{proposition}
\label{p:alpha-incr}
 $(\alpha_k)_{k\in\mathbb{N}_+}$ is a strictly
  increasing sequence in $[\tfrac{1}{2},1)$ while
$(\beta_k)_{k\in\mathbb{N}_+}$ is a strictly decreasing sequence in
    $(1,2]$.
Moreover $1 + \tfrac{1}{k} \le \beta_k \le 1 + \tfrac{1}{\sqrt{k}}$
hence $\beta_k$ and $\alpha_k$ converge to $1$ when $k\to\infty$.
\end{proposition}

\begin{theorem}
\label{t:limits}
For $k\ge 1$ and $j\ge 0$, the following limits exist and have the
given values:
\begin{align*}
\lim_{n\to\infty} \tfrac{1}{n}F_k^j(n) &= \alpha_k^j \\
\lim_{n\to\infty} \tfrac{1}{n}L_k^j(n) &= \beta_k^j
\end{align*}
Said otherwise, when $n\to\infty$ we have
$F_k^j(n) \sim \alpha_k^j\,n$ and $L_k^j(n) \sim \beta_k^j\, n$.
\end{theorem}

As a side remark, we are now able to provide asymptotic equivalents
for $\alpha_k$ and $\beta_k$, instead of just the previous bounds
of Proposition~\ref{p:alpha-incr}. Actually, the equivalent for
$\beta_k$ was already proposed by Selmer~\cite[Eq.~(4.4)]{Selmer56}
but without justifications. See also Dilcher~\cite[Lem.~3]{Dilcher1993}.

\begin{proposition}
\label{p:asympt}
When $k$ converges to $+\infty$,
\begin{align*}
\alpha_k &= 1 - \frac{\ln(k)}{k} + o\left(\frac{\ln(k)}{k}\right) \\
\beta_k &= 1 + \frac{\ln(k)}{k} + o\left(\frac{\ln(k)}{k}\right)
\end{align*}
\end{proposition}
\begin{proof}
Let $k\ge 1$.
$\alpha_k$ can be expressed as the solution of the equation
$x=\sqrt[k]{1-x}$ in the interval $(0,1)$. Even more,
it is an attracting fixed point of
the function $f(x)=\sqrt[k]{1-x}$, with the whole interval $(0,1)$ as
basin of attraction. It is hence sufficient to iterate this function
to get approximations of $\alpha_k$. Here, these approximations will
alternate on each side of $\alpha_k$. We choose
$u_0 = 1{-}\tfrac{1}{e}$ as initial point, since this cancels some
constants later,
then the sequence $u_{n+1}=f(u_n)$ converges to $\alpha_k$.
Actually,
$u_2$ and $u_3$ already have asymptotic expansions that suit our
needs. Indeed, we obtain $u_2 < \alpha_k < u_3$ (as soon as $k\ge 3$) and
\begin{align*}
u_2 &= f^2\left(1-\frac{1}{e}\right) = 1 - \frac{\ln(k)}{k} +
 o\left(\frac{1}{k}\right) \\
u_3 &= f(u_2) = 1 - \frac{\ln(k)}{k} + \frac{\ln(\ln(k))}{k} +
 o\left(\frac{1}{k}\right).
\end{align*}
Now, the expansion of $\beta_k$ is obtained by inversing the one of
$\alpha_k$.
\end{proof}

Let us study now the complex zeros of polynomials $P_k$ and $Q_k$.
Since $P_k$ has no common zero with its derivative,
it admits $k$
distinct complex zeros, including $\alpha_k$ and a real negative
zero when $k$ is even. The $k$ zeros of $Q_k$ are the inverse of the
zeros $P_k$, they are also distinct
and include $\beta_k$ and also a real negative zero when $k$ is even.
Moreover, $\beta_k$ is the dominant zero of $Q_k$:

\begin{proposition}
\label{p:betamax}
Let $z\in\mathbb{C}$ be a zero of $Q_k$ different from $\beta_k$. Then
$|z|<\beta_k$.
\end{proposition}
\begin{proof}
Via triangle inequality, $|(z-1)+1|\le |z-1|+|1|$
hence $|z|-1 \le |z-1|$. Moreover this is an equality only when
$z-1\in \mathbb{R_+}$. But $\beta_k$ is the unique zero of $Q_k$ which
is a positive real number. So here $|z|-1 < |z-1|$.

Now, if $\beta_k \le |z|$ we could derive the following
contradiction:
\[
1 = (\beta_k -1).\beta_k^{k-1} \le (|z|-1).|z|^{k-1} < |z-1|.|z|^{k-1}
 = |z^k-z^{k-1}| = 1.
\]
\end{proof}

For the same kind of reasons, when $Q_k$ admits a real negative zero
(i.e., when $k$ is even), then the absolute value of this zero is
strictly smaller than the modulus of the other zeros.

\begin{proposition}
For two zeros $z$ and $t$ of $Q_k$, $|z|=|t|$ if and only if
$\RE(z)=\RE(t)$, which happens exactly when $z$ and $t$ are equal or
complex conjugates. Similarly, $|z|<|t|$ whenever $\RE(z)<\RE(t)$.
\end{proposition}
\begin{proof}
Assume $|z|=|t|$.
\[
|z-1|.|z|^{k-1} = |(z-1)z^{k-1}| = 1 = |(t-1)t^{k-1}| = |t-1|.|t|^{k-1}.
\]
This simplifies into $|z-1|=|t-1|$ here,
since $0$ is not a zero of $Q_k$ and hence $|z|=|t|\neq 0$.
Let us write $z=a{+}ib$ and $t=c{+}id$.
We deduce $a^2+b^2=c^2+d^2$ and $(a-1)^2+b^2=(c-1)^2+d^2$.
Developing and subtracting these two equations leads to $a=c$, i.e.,
$\RE(z)=\RE(t)$, then $b=\pm d$, i.e., $t=z$ or $t=\overline{z}$.
Similarly, when $|z|<|t|$ we obtain $|z|^2<|t|^2$  and
$|t-1|^2<|z-1|^2$ and then $a<c$, i.e., $\RE(z)<\RE(t)$.
The reciprocal statements are obtained by totality of the order on
$\mathbb{R}$.
\end{proof}

\begin{definition}
For $k\ge 1$, let us name $r_{k,0} \cdots r_{k,k-1}$ the $k$ complex
zeros of $Q_k$, in decreasing lexicographic order of their Cartesian
coordinates.
\end{definition}

Thanks to the previous properties, this choice of ordering implies
that this sequence $r_{k,i}$ of zeros have decreasing modulus, with
equality of modulus exactly when a non-real zero is followed by its
complex conjugate. More precisely:
\begin{itemize}
\item the positive zero comes first: $r_{k,0} = \beta_k$;
\item the possible negative zero comes last:
  $r_{k,k-1}\in\mathbb{R}_{-}$ when $k$ is even;
\item the pairs of complex conjugate zeros come in between:
for $1< 2p < k$, we have $r_{k,2p} = \overline{r_{k,2p-1}}$.
\end{itemize}

Let us state some algebraic equations between these zeros.

\begin{proposition}
\label{p:sumprod}
For $k>1$, the product of the $r_{k,i}$ numbers is $(-1)^{k-1}$ while their
sum is 1, as well as their first Newton sums:
\[
\sum_{i=0}^{k{-}1} r_{k,i}^p = 1 \qquad \text{for}\ 1\le p < k.
\]
\end{proposition}
\begin{proof}
Since the $r_{k,i}$ numbers are the zeros of the polynomial
$Q_k(X) = X^k-X^{k-1}-1$, the Vieta's formulas directly give
their sum and product from the coefficients of $Q_k$, moreover
all the other elementary symmetric polynomials $\sigma_p$ of these roots
are null for $1<p<k$. Thanks to the Newton identities,
this implies that the Newton sums $\sum_i r_{k,i}^p$ are always $1$
for $1\le p < k$.
\end{proof}

For studying in Section~\ref{s:distF} the maximal distance between
$F_k$ and its linear equivalent $\alpha_k\,n$,
it will be essential to know whether the secondary zeros of $Q_k$
may reach or even exceed 1 in modulus (by secondary, we mean different
from the dominant zero $\beta_k$).

\begin{proposition}\label{p:secondary}
Let $k\ge 2$.
The secondary zeros of $Q_k$ of maximal modulus are $r_{k,1}$
and $r_{k,2} = \overline{r_{k,1}}$ (or solely $r_{k,1}$ when $k=2$).
Moreover the modulus $|r_{k,1}|$ is as follow:
\begin{itemize}
\item $|r_{k,1}|<1$ when $k\in\{2,3,4\}$;
\item $|r_{k,1}|=1$ when $k=5$;
\item $|r_{k,1}|>1$ when $k\ge 6$.
\end{itemize}
\end{proposition}
\begin{proof}
Our choice of ordering for the zeros $r_{k,i}$ implies that
the secondary zeros of maximal modulus are $r_{k,1}$ and its
conjugate $r_{k,2}$ (or solely $r_{k,1}$ when $k=2$).

First, $\beta_2=\varphi\approx 1.618$ and $\beta_3\approx 1.465$ and
$\beta_4 \approx 1.380$ are
well-known
Pisot numbers~\cite{DufresnoyPisot55} and $Q_2$, $Q_3$, $Q_4$
are their minimal polynomials, hence their other zeros
are all of modulus strictly below 1.

Now, the case $k=5$ is quite
particular, since $Q_5$ can be factorized as $(X^2-X+1)(X^3-X-1)$.
On one side, the zeros of $X^2-X+1$ are
$\tfrac{1+i\sqrt 3}{2} = e^{i\pi/3}$ and its conjugate,
both of modulus 1.
On the other side, $\beta_5$ is the real zero of $X^3-X-1$ and it is
hence the well-known Plastic Ratio $\approx 1.324718$, which is the
smallest Pisot number, as shown by Siegel~\cite{Siegel1944}.
The two remaining complex zeros of
$X^3-X-1$ are of modulus $<1$. Finally $r_{k,1}=e^{i\pi/3}$.

For $k\ge 6$, let us name $M_k$ the minimal polynomial of $\beta_k$.
This polynomial $M_k$ is monic, it has integer coefficients and it
divides $Q_k$, so the zeros of $M_k$ are also zeros of $Q_k$.
The degree of $M_k$ is at least 2, otherwise $\beta_k$ would be an integer.
So we can consider the largest modulus among the zeros of
$M_k$ distinct from $\beta_k$ and compare this modulus with 1.
\begin{itemize}
\item If this modulus is strictly less than 1, then $\beta_k$ is a
  Pisot number, but this case is impossible here, since $\beta_k$ is
  strictly less than the smallest Pisot $\beta_5$, see
  Siegel~\cite{Siegel1944}.
\item If this modulus is exactly 1, then $\beta_k$ is a Salem
  number~\cite{salem1963algebraic}, but this case is also impossible
  here. For instance, a Salem number is known to also have its
  inverse as zero of its minimal polynomial (which is self-reciprocal)
  while here neither $Q_k$ nor $M_k$ have zeros in $[0,1]$.
\item Finally, this modulus can only be strictly more that 1. This
  proves the existence of a zero of $M_k$ distinct from $\beta_k$ with
  modulus $>1$,
  and this zero is also a zero of $Q_k$. Considering the ordering of
  zeros of $Q_k$ chosen above, we can conclude $|r_{k,1}|>1$.
\end{itemize}
\end{proof}

As a complement to the previous proof, it is easy to show that $Q_k$
admits a zero $z$ of modulus 1 only when
$|z-1|=|(z-1).z^{k-1}|=1$, and having both $|z|=1$ and $|z-1|=1$ can
only occur for $z=e^{\pm i\pi/3}$, the zeros of $X^2-X+1$ already
encountered for $Q_5$.
In this case, $z-1 = z^2$ and the equation $Q_k(z)=0$ amounts to
$z^{k+1}=1$ hence $e^{\pm i(k+1)\pi/3}=1$ hence $k = \modd{5}{6}$.
When $k = \modd{5}{6}$, there exists a polynomial
$R_k\in\mathbb{Z}[X]$ such that $Q_k=(X^2-X+1)R_k(X)$. We have
already seen $R_5=X^3-X-1$ above; the other $R_k$ can be obtained via the
recursive equation $R_{k+6}=R_k+(X^5-X^3-X^2+1)X^{k-1}$.
Actually, this is the only possible factorization of $Q_k$ over
$\mathbb{Z}[X]$, or equivalently over $\mathbb{Q}[X]$. Indeed,
the work of Selmer~\cite{Selmer56} implies that
when $k\neq \modd{5}{6}$ the polynomial $Q_k$ is irreducible over
$\mathbb{Z}[X]$. This also
holds for the reciprocal polynomial $P_k$.
For a simpler proof, see Conrad~\cite{Conrad}. In this case, $Q_k$ is
hence the minimal polynomial of $\beta_k$.
This work of Selmer~\cite{Selmer56}
also implies that for $k=\modd{5}{6}$, the factor $R_k$ of $Q_k$ is always
irreducible over $\mathbb{Z}[X]$, and is hence the minimal polynomial
of $\beta_k$ in this case.

\smallskip

\section{Algebraic expressions of Fibonacci-like sequences}
\label{s:Aalg}

For the next section, it is crucial to provide precise algebraic
expressions for Fibonacci-like sequences $(\A{k}{n})$. In particular
we will distinguish here the main exponential behavior of $(\A{k}{n})$
from residual parts.

Recall from section~\ref{s:poly} that $Q_k(X)=X^k-X^{k-1}-1$ and
$\beta_k\in(1,2]$ is the unique positive zero of $Q_k$, as well as its
dominant zero. Also recall that the $k$ complex zeros $r_{k,i}$
of $Q_k$ are simple and hence distinct. In our notation, $r_{k,0} = \beta_k$.

\begin{proposition}\label{p:betaA}
For all $k\ge 1$ and $n\ge 0$, $\beta_k^n \le \A{k}{n}$
\end{proposition}
\begin{proof}
We proceed by strong induction over $n$.
Let $n\ge 0$ and assume $\beta_k^m \le \A{k}{m}$ for all $0\le m<n$.
Since $Q_k(\beta_k) = 0$, we always have:
\[
\beta_k^n = \beta_k^{n-k}.\beta_k^k = \beta_k^{n-k}.(\beta_k^{k-1}+1)
        = \beta_k^{n-1} + \beta_k^{n-k}.
\]
If $n=0$ then $\beta_k^0 = 1 = \A{k}{0}$. Assume now $n\neq 0$.
Note that $\beta_k^{n-k} \le \beta_k^{n\ominus k}$. Indeed,
the exponents are either equal (when $n\ge k$) or negative on the left
and $0$ on the right (when $n < k$). So:
\[
\beta_k^n = \beta_k^{n-1} + \beta_k^{n-k}
        \le \beta_k^{n-1} + \beta_k^{n\ominus k}
        \le \A{k}{n-1} + \A{k}{n\ominus k} = \A{k}{n}
\]
by induction hypothesis in $n{-}1$ and $n{\ominus}k$ and
Equation~\eqref{Agenrec}.
\end{proof}

\begin{proposition}\label{p:Alincomb}
Let us name $c_{k,i}$ the complex quantity
$r_{k,i}^k\,(k\,r_{k,i}{-}(k{-}1))^{-1}$.
All these coefficients $c_{k,i}$ differ from $0$ and:
\begin{equation}
\label{Alincomb}
   \A{k}{n} = \sum_{i=0}^{k-1} c_{k,i}\,r_{k,i}^n \qquad \text{
     for all}\ n\ge 0.
\end{equation}
\end{proposition}

\begin{proof}
See Dilcher~\cite[Eq.~(2.8)]{Dilcher1993}.
We present here a standalone proof based on Lagrange interpolation
polynomials.

Note first that all real roots of $Q_k$ are either higher than 1 or
strictly negative, hence $r_{k,i}\neq 0$ and
$(k\,r_{k,i}{-}(k{-}1))\neq 0$
hence $c_{k,i}$ is a finite complex different from $0$.
The case $k=1$ is obvious since $r_{1,0}=2$ and $c_{1,0}=1$ and $A_{1,n}=2^n$.
We now assume $k>1$ for the rest of this proof.
Since $k>1$, the $A_{k,n}$ sequence can be extended downward on all
negative indices $n$, by turning the recursive rule of Definition~\ref{Agenrec}
into $A_{k,n} = A_{k,n+k}-A_{k,n+k-1}$. In particular $A_{k,n} = 1$
when $-(k-1) \le n \le 0$ and $A_{k,n} = 0$ when $-(2k{-}2)\le n \le -k$.
Here, it is convenient to consider a shifted
version of $A_{k,n}$, namely $\tilde{A}_{k,n} := \A{k}{(n-2k+2)}$.
This sequence follows the same recursive rule as $A_{k,n}$, but its
$k$ first values are $0, \ldots, 0, 1$.
The corresponding coefficients are now
$\tilde{c}_{k,i} := c_{k,i}/r_{k,i}^{2k-2}$.
So let us abbreviate
$\Sigma_n := \sum_{i=0}^{k-1} \tilde{c}_{k,i}\, r_{k,i}^n $
and show that $\tilde{A}_{k,n} = \Sigma_n$ for all $n\ge 0$.
Since $r_{k,i}^k = r_{k,i}^{k-1}+1$,
these sums $\Sigma_n$ also follow the recursive rule
$\Sigma_n = \Sigma_{n-1}+\Sigma_{n-k}$, here for all $n\in\mathbb{Z}$.
So it suffices to show $\tilde{A}_{k,n} = \Sigma_n$ for the initial values
$0\le n \le k-1$, i.e., establish $\Sigma_n = 0$ for $0 \le n < k-1$
and $\Sigma_{k-1} = 1$.

For this purpose, we consider the following interpolation polynomial
for the zeros $r_{k,i}$ (which are already known to be all distinct):
\[
\mathcal{L}(X) = -1 + \sum_{i=0}^{k-1}
\prod\limits_{\substack {0\le j< k \\ j\ne i}}
\frac{X - r_{k,j}}{r_{k,i} - r_{k,j}}.
\]
This polynomial has a degree strictly less than $k$ while admitting the
$k$ distinct numbers $r_{k,i}$ as zeros. It is hence the null polynomial.
Let us reformulate now $\mathcal{L}(X)$ in terms of sums
$\Sigma_n$. First, for a given $i<k$, the numerator inside $\mathcal{L}$
satisfies
\[
\prod_{j\neq i} (X - r_{k,j}) = X^{k-1} + \sum_{m=0}^{k-2}
r_{k,i}^{-m-1} X^m.
\]
Indeed, multiplying by $(X - r_{k,i})$ on both sides yields
$Q_k(X) = X^k{-}X^{k-1}{-}1$ (note that $r_{k,i}^{-(k-1)} = r_{k,i}-1$
since $r_{k,i}$ is a zero of $Q_k$).
Secondly, by differentiating the polynomial
$Q_k(X) = \prod_i (X-r_{k,i})$ and evaluating the
result at each $r_{k,i}$, we can express the denominators in
$\mathcal{L}$ as
\[
\prod_{j\neq i} (r_{k,i} - r_{k,j}) = Q'_k(r_{k,i}) =
 k\,r_{k,i}^{k-1} - (k{-}1) r_{k,i}^{k-2} = 1 / \tilde{c}_{k,i}
\]
(alternatively, this last equation can also be obtained by instantiating
$X=r_{k,i}$ in the previous equation).
Finally, after substituting these equations in $\mathcal{L}$ and
regrouping the sums, we obtain
\[
\mathcal{L}(X) =
  \Sigma_0 X^{k-1} + \left(\sum_{m=1}^{k-2} \Sigma_{(-m-1)} X^m \right) +
  \Sigma_{(-1)} -1
\]
Since $\mathcal{L}(X) = 0$, all the coefficients of $\mathcal{L}$
are null, hence:
\begin{itemize}
\item $\Sigma_0 = 0$
\item For $1\le m \le k{-}2$,
  $\Sigma_{(-m-1)} = 0 = \Sigma_{k-m-1} - \Sigma_{k-m-2}$.
  These equations, combined with $\Sigma_0 = 0$, allow to
  successively obtain $\Sigma_n = 0$ for all $0\le n \le k{-}2$.
\item $\Sigma_{(-1)} = 1 = \Sigma_{k-1} - \Sigma_{k-2}$.
  Since we just proved $\Sigma_{k-2} = 0$, then $\Sigma_{k-1} = 1$.
\end{itemize}
This allows to conclude $\tilde{A}_{k,n} = \Sigma_{n}$ for all $n\ge 0$,
since these sequences have the same $k$ first values and the same
recursive rule. Actually, this equality even holds for all $n\in\mathbb{Z}$,
as a downward induction can handle the negative $n$, but we will not
use this fact. Finally, Equation~\eqref{Alincomb} is deduced by
considering $A_{k,n} = \tilde{A}_{k,(n+2k-2)} = \Sigma_{(n+2k-2)}$.
\end{proof}

Alternatively, the previous proof can also be done by inverting
the Vandermonde matrix of the roots, and multiply this inverse by the
vector of initial values of $(A_{k,n})$.
Since the general expression of this Vandermonde inverse is pretty
huge, it is crucial here again to shift $(A_{k,n})$ downward:
the initial values $0, \ldots, 0, 1$ of $(\tilde{A}_{k,n})$ allow to
use only
the last line of the Vandermonde inverse, which is far more tractable,
and yields the same expression as above.

Note that for $k=2$, this shifted sequence $(\tilde{A}_{2,n})$
is actually the usual definition of the Fibonacci
numbers, with initial values $0$ $1$, and the formula
$\tilde{A}_{2,n} = \sum_i \tilde{c}_{2,i}\,r_{2,i}^n$ is here the well-known
Binet formula, since here $r_{2,0} = \beta_2 = \varphi$ and
$r_{2,1} = 1-\varphi = \tfrac{-1}{\varphi} $ and
$\tilde{c}_{2,i}=\tfrac{\pm 1}{\sqrt{5}}$.
The Equation~\eqref{Alincomb} for $\A{2}{n}$ is less elegant, with
$c_{2,i} = \tfrac{1}{10}(5\pm 3\sqrt{5})$, but the choice of a
sequence without $0$ nor duplicated $1$ helps in
section~\ref{s:numrep}.

\begin{corollary}
$c_{k,0}$ is a real greater or equal to $1$ and
\[
\lim_{n\to\infty} \frac{\A{k}{n}}{\beta_k^n} = c_{k,0}
\quad
\text{and}
\quad
\lim_{n\to\infty} \frac{\A{k}{n+1}}{\A{k}{n}} = \beta_k.
\]
\end{corollary}
\begin{proof}
Recall from Proposition~\ref{p:betamax} that
$|r_{k,i}| < \beta_k $ for $1\le i< k$. The first limit above is
hence a direct consequence of Equation~\eqref{Alincomb}, all
terms $c_{k,i} (r_{k,i}/\beta_k)^n$ for $i\neq 0$ tend to 0 when $n$ grows.

Now, recall from Proposition~\ref{p:betaA} that
$\beta_k^n \le \A{k}{n}$.
As a finite limit of real numbers greater or equal to $1$,
$c_{k,0}$ is hence a real greater or equal to $1$.

Finally
\[
\frac{\A{k}{n+1}}{\A{k}{n}} =
\beta_k \cdot
\frac{\A{k}{n+1}}{\beta_k^{n+1}}\cdot\frac{\beta_k^n}{\A{k}{n}}
\]
hence this ratio converges to $\beta_k\,c_{k,0}\,c_{k,0}^{-1} = \beta_k$.
\end{proof}

In practice, the first values of $c_{k,0}$ are
$c_{1,0} = 1$ and $c_{2,0} \approx 1.17$ and
$c_{3,0} \approx 1.31$ and $c_{4,0} \approx 1.43$.
Despite these moderate initial values, $c_{k,0}$ tends to $+\infty$
when $k$ tends to $+\infty$, since the definition of $c_{k,0}$
combined with Proposition~\ref{p:asympt} lead to an asymptotic
equivalent of $k/(\ln(k))^2$.

As a side note, the coefficients $c_{k,i}$ are the $k$
zeros of a polynomial in $\mathbb{Z}[X]$ of degree $k$, with leading
coefficient $k^k+(k{-}1)^{k-1}$ when $k>1$ (or $1$ when $k=1$) and
constant coefficient $(-1)$. This fact is also true for
the coefficients $\tilde{c}_{k,i}$. Here are the first such polynomials:
\begin{itemize}
\item $X-1$ for $c_{1,0}=\tilde{c}_{1,0}=1$.
\item $5X^2-5X-1$ for $c_{2,i}$ and $5X^2-1$ for $\tilde{c}_{2,i}$.
\item $31X^3-31X^2-12X-1$ for $c_{3,i}$ and $31X^3+X-1$ for $\tilde{c}_{3,i}$.
\item $283X^4-283X^3-162X^2-24X-1$ for $c_{4,i}$ and $283X^4+6X^2+X-1$
  for $\tilde{c}_{4,i}$.
\end{itemize}
For the moment, we lack a general expression directly giving these
polynomials, but their existence is proved in
Appendix~\ref{s:polycoeff}.
Some more coefficients can be easily determined,
for instance the ones of degree $k{-}1$ since
$\sum_i c_{k,i} = 1$ (see Equation~\eqref{Alincomb} when $n=0$)
and similarly $\sum_i \tilde{c}_{k,i} = 0$.
For more on the sequence $5,\ 31,\ 283,\ 3381, \ldots k^k{+}(k-1)^{k-1}$,
see OEIS~\seqnum{A056788}.

Recall now that $\alpha_k = 1/\beta_k = 1/r_{k,0}$.

\begin{proposition}\label{p:Adiff}
For $k\ge 2$ and $0\le i< k$, let us name $d_{k,i}$ the complex quantity
$c_{k,i}\,(r_{k,i}^{-1}-\alpha_k)$. Hence $d_{k,i}=0$ iff $i=0$.
Then for all $n\ge 0$,
\[
\A{k}{n\ominus 1} - \alpha_k\,\A{k}{n} =
\sum_{i=1}^{k-1} d_{k,i}\,r_{k,i}^n.
\]
\end{proposition}
\begin{proof}
For $n>0$, $n{\ominus}1 = n{-}1$, and we apply twice
Proposition~\ref{p:Alincomb}:
\[
\A{k}{n-1} - \alpha_k\,\A{k}{n}
= \sum_{i=0}^{k-1} c_{k,i}\,r_{k,i}^{n-1}-\alpha_k\,\sum_{i=0}^{k-1} c_{k,i}\,r_{k,i}^n
= \sum_{i=0}^{k-1} c_{k,i}\,(\frac{1}{r_{k,i}}-\alpha_k)\,r_{k,i}^n
= \sum_{i=0}^{k-1} d_{k,i}\,r_{k,i}^n
\]
and the first term of this sum can be ignored since $d_{k,0} = 0$.

For $n=0$, $\A{k}{0} -\alpha_k\,\A{k}{0} = 1-\alpha_k$ while
we have (reusing the notations of the proof of Proposition~\ref{p:Alincomb}):
\[
\sum_{i=1}^{k-1} d_{k,i} = \sum_{i=0}^{k-1} d_{k,i} =
\sum_{i=0}^{k-1} \tilde{c}_{k,i}\,r_{k,i}^{2k-3} -
  \alpha_k\,\sum_{i=0}^{k-1} c_{k,i}\,r_{k,i}^0
= \tilde{A}_{k,(2k-3)} - \alpha_k\,\A{k}{0}
= 1 - \alpha_k.
\]
\end{proof}

\begin{lemma}\label{lemtech}
Let $\theta,\rho\in\mathbb{R}$ and assume $\sin \theta \neq 0$. If the
sequence $(a_n)_{n\in\mathbb{N}}$ is defined by $a_n = \cos(n\theta+\rho)$,
then there exists a subsequence of $(a_n)$ always greater or equal to $1/2$.
Similarly there exists a subsequence of $(a_n)$ always smaller or equal
to $-1/2$.
\end{lemma}
\begin{proof}
If $\theta/2\pi$ is rational, the sequence $(a_n)$ is periodic and
we find infinitely many $n$ such that $a_n\ge 1/2$.
More precisely, we write $\theta/2\pi = p/q$ with
$p,q\in \mathbb{Z}$ and $q>0$ and $p,q$ relatively prime.
Moreover, $q$ cannot be 1 nor 2, otherwise $\sin \theta = 0$.
By Bezout, there exist integers $u,v$ such that $u p+v q = 1$.
By Euclidean division extended to real numbers,
there exists an integer $k$ and a real $0\le r<1/q$ such that
$(-\rho/2\pi+1/6) = k(1/q)+r$.
Then $\cos((k u)\theta+\rho) = cos(\pi/3-2\pi r) \ge 1/2$, since
$0\le 2\pi r < 2\pi/q \le 2\pi/3$.
A suitable subsequence is hence all the indices $n = \modd{k u}{q}$
that are positive.

We assume now that $\theta/2\pi$ is irrational. Let $N$ be a natural number.
The Kronecker
approximation theorem provides two integers $p,q$ such that $q>0$ and
\[|(\theta/2\pi)q - p + (N\theta+\rho)/2\pi|<1/6.\]
This implies that $\cos((N{+}q)\theta+\rho) \ge \cos(\pi/3) = 1/2$.
The interval $[1/2,1]$ is hence visited infinitely often by the
sequence $(a_n)$, ensuring the existence of the desired subsequence.

Finally, for finding a subsequence always smaller or equal to $-1/2$, it is
sufficient to use the first part of the statement with the same
$\theta$ but a shifted $\rho+\pi$, changing the sign of all cosine.
\end{proof}

\begin{proposition}\label{p:supAdiff}
For $k\ge 6$,
$\sup_{n\in\mathbb{N}} |\A{k}{n\ominus 1}-\alpha_k\,\A{k}{n}| = +\infty$.
More precisely, for a given $k\ge 6$ and the constant
$C = |d_{k,1}|/2 > 0$, there exist two strictly increasing sequences of
natural numbers $(u_n)$ and $(u'_n)$ such that for all $n\ge 0$,
\[
\A{k}{u_n\ominus 1}-\alpha_k\,\A{k}{u_n} > C |r_{k,1}|^{u_n}
\]
and
\[
\A{k}{u'_n\ominus 1}-\alpha_k\,\A{k}{u'_n} < -C |r_{k,1}|^{u'_n}.
\]
\end{proposition}
\begin{proof}
Let $k\ge 6$. Recall from Proposition~\ref{p:secondary} that in this
case $|r_{k,1}|>1$ and $r_{k,2}=\overline{r_{k,1}}$ and for
$3\le i< k$ we have $|r_{k,i}|\le|r_{k,3}|<|r_{k,1}|$.
Also note that $d_{k,2} = \overline{d_{k,1}}$ just as
$c_{k,2} = \overline{c_{k,1}}$, thanks to their definitions.
Now, we apply Proposition~\ref{p:Adiff} and regroup the two first terms:
\[
\A{k}{n\ominus 1} - \alpha_k\,\A{k}{n}
= \sum_{i=1}^{k-1} d_{k,i}\,r_{k,i}^n
= 2\,\RE(d_{k,1}\,r_{k,1}^n) + \sum_{i=3}^{k-1} d_{k,i}\,r_{k,i}^n
\]
Since $|r_{k,i}|\le|r_{k,3}|$ for all $i\ge 3$, we
pose $C_3 = \sum_{3\le i} |d_{k,i}| > 0$ and get
\[
\left|\sum_{i=3}^{k-1}d_{k,i}\,r_{k,i}^n\right| \le C_3 |r_{k,3}|^n
\]
and hence
\[
 2\,\RE(d_{k,1}\,r_{k,1}^n) - C_3 |r_{k,3}|^n
 \le \A{k}{n\ominus 1} - \alpha_k\,\A{k}{n}
 \le 2\,\RE(d_{k,1}\,r_{k,1}^n) + C_3 |r_{k,3}|^n.
\]

By writing $d_{k,1}$ as $|d_{k,1}|e^{i\rho}$ and $r_{k,1}$
as $|r_{k,1}|e^{i\theta}$, we get
\[
\RE(d_{k,1}\,r_{k,1}^n) = |d_{k,1}|\cdot|r_{k,1}|^n\cdot\cos(n\theta+\rho).
\]
Since $k>2$, $r_{k,1}\notin\mathbb{R}$ hence $\sin \theta\neq 0$,
and the previous Lemma~\ref{lemtech} provides a
strictly increasing sequence $(v_n)$ such
that $|\cos(v_n \theta+\rho)| \ge 1/2$ and hence
\[
(|d_{k,1}| - C_3\cdot\chi^{v_n})\cdot|r_{k,1}|^{v_n}
\le
\A{k}{v_n\ominus 1} - \alpha_k\,\A{k}{v_n}
\]
where $\chi = |r_{k,3}|/|r_{k,1}| < 1$.
Recall that $C = |d_{k,1}|/2 > 0$.
We can find a shift $N$ large enough for having
$\chi^{v_{(n+N)}} < C/C_3$ for all $n\ge 0$. Taking
$u_n = v_{(n+N)}$ ensures the first desired inequality.

On the other side, the Lemma~\ref{lemtech} also provides a
strictly increasing sequence $(v'_n)$ such that
$|\cos(v'_n \theta+\rho)| \le -1/2$, leading to
\[
\A{k}{v'_n\ominus 1} - \alpha_k\,\A{k}{v'_n} \le
-(|d_{k,1}| - C_3\cdot\chi^{v'_n})\cdot|r_{k,1}|^{v'_n}.
\]
and $u'_n = v'_{(n+N')}$ is a suitable sequence as soon as the shift
$N'$ is large enough for having $\chi^{v'_{(n+N)}} < C/C_3$ for all $n\ge 0$.
\end{proof}

\section{Discrepancy: maximal distance to the linear equivalent}
\label{s:distF}

From Theorem~\ref{t:limits}, we know that for any $k\ge 1$,
$F_k(n)$ admits $\alpha_k\,n$ as linear equivalent. We investigate
here the distance between $F_k$ and this linear equivalent.

\begin{definition}
For $k\ge 1$, we note $\delta_k(n) = F_k(n){-}\alpha_k\,n$ and
we call \emph{discrepancy} the quantity
$\Delta_k = \sup_{n\in\mathbb{N}} |\delta_k(n)|$.
\end{definition}

We now retrieve that the discrepancy $\Delta_k$ is finite
(and quite small) for $k\le 4$ and infinite for $k\ge 5$.
These results are immediate consequences of Dilcher
theorems~\cite{Dilcher1993}. We are more precise here,
and state that for $k\ge 5$ we have both\footnote{
For the rest of this article, the supremum and infimum expressions
are written without domain when it is $\mathbb{N}$, i.e., the full
definition domain of the function $\delta_k$.}
$\sup_n \delta_k(n) = +\infty$ and $\inf_n \delta_k(n) = -\infty$.
Moreover for $k=3,4$ we provide much better bounds than Dilcher,
whose actual study of $|F_3(n){-}(n{-}1)\alpha_3{-}1|$ is quite awkward
compared to $|\delta_3(n)|$, same for $k=4$.
The first version of this document proved just that
$|\delta_3(n)|<1$ and $|\delta_4(n)|<2$, but we can now be much more
precise, thanks to new methods inspired in part by
Shallit~\cite{Shallit2025}.

\subsection{The general case}

In the general case $k\ge 6$, Proposition~\ref{p:supAdiff}
implies that $\delta_k$ diverges at least for some subsequence
of the $\A{k}{n}$ numbers. More precisely:

\begin{theorem}\label{t:supdelta}
For $k\ge 6$, $\Delta_k = +\infty$ and even
$\sup_n \delta_k(n) = +\infty$ and $\inf_n \delta_k(n) = -\infty$.
More precisely, for a given $k\ge 6$,
there exist a constant $C\in\mathbb{R}_+^*$
and two subsequences $(u_n)$ and $(u'_n)$ of the sequence
 $(\A{k}{n})_{n\in\mathbb{N}}$
such that for all $n\ge 0$,
\[
\delta_k(u_n) > C\,(u_n)^a
\]
and
\[
\delta_k(u'_n) < -C\,(u'_n)^a
\]
where $a = \ln |r_{k,1}| /\ln \beta_k$ and hence $0<a<1$.
\end{theorem}
\begin{proof}
Let $k\ge 6$.
Recall that $r_{k,1}$ is one of the zeros of the polynomial $Q_k$,
of maximal modulus among the zeros distinct from the positive zero
$\beta_k$. And from Proposition~\ref{p:secondary}, we know that here
$|r_{k,1}|>1$.

Let us call $C_0$ and $v_n$ and $v'_n$ the constant and sequences
provided by Proposition~\ref{p:supAdiff} for this $k$, and combine
it with Proposition~\ref{p:FA}:
\[
\delta_k(\A{k}{v_n})=F_k(\A{k}{v_n})-\alpha_k\,\A{k}{v_n}
               =\A{k}{v_n\ominus 1}-\alpha_k\,\A{k}{v_n}
               > C_0\,|r_{k,1}|^{v_n}
               = C_0\,(\beta_k^{v_n})^a.
\]
Recall from Proposition~\ref{p:Alincomb} that
$\lim_{n\to\infty} \A{k}{n}/\beta_k^n = c_{k,0} > 0$.
So we can find $N$ such that for all $m\ge N$,
$\A{k}{m} < 2 c_{k,0} \beta_k^m $
and hence
$(\beta_k^m)^a > (2 c_{k,0})^{-a} (\A{k}{m})^a$.
We can now choose $C = C_0\,(2 c_{k,0})^{-a}$ and $u_n = \A{k}{v_{(n+N)}}$
and put everything together (note that $v_{(n+N)}\ge N$):
\[
\delta_k(u_n) > C_0\,(\beta_k^{v_{(n+N)}})^a
> C_0\,(2 c_{k,0})^{-a}\,(\A{k}{v_{(n+N)}})^a
= C\,(u_n)^a.
\]
On the other side, we obtain similarly that
$\delta_k(\A{k}{v'_n}) < -C_0\,(\beta_k^{v'_n})^a$ and
with the same constant $C$ and shift $N$ as before, we pose
$u'_n = \A{k}{v'_{(n+N)}}$ and get $\delta_k(u'_n) < -C\,(u'_n)^a$.
\end{proof}

Experimentally, the exponent is $a \approx 0.1287$ for $k=6$ and
$a \approx 0.2218$ for $k=7$ and appears to tend to $1$ when $k$ grows.

For the remaining results of this section, we need to express
$\delta_k(n)$ thanks to the $k$-decomposition of $n$.

\begin{proposition}\label{p:deltadecomp}
Let $k\ge 2$ and $n\ge 0$.
\[
\delta_k(n) =
\sum_{q\in D_k(n)} \sum_{i=1}^{k-1} d_{k,i}\,r_{k,i}^q =
\sum_{i=1}^{k-1} \left(d_{k,i}\sum_{\!q\in D_k(n)\!} r_{k,i}^q\right)
\]
where $d_{k,i}$ comes from Proposition~\ref{p:Adiff} and $D_k(n)$ from
Theorem~\ref{t:zeck}.
\end{proposition}
\begin{proof}
\begin{align*}
\delta_k(n) &= F_k(n) - \alpha_k\,n & \\
  &= \Sigma_k(D_k(n) \gg_+ 1) - \alpha_k\,n
     \qquad & \text{(by Thm.~\ref{t:Fshift})}\\
  &= \sum_{q\in D_k(n)} \A{k}{q\ominus 1} -
     \alpha_k \sum_{\!q\in D_k(n)\!} \A{k}{q} & \\
  &= \sum_{q\in D_k(n)} (\A{k}{q\ominus 1} - \alpha_k\,\A{k}{q}) & \\
  &= \sum_{q\in D_k(n)} \sum_{i=1}^{k-1} d_{k,i}\,r_{k,i}^q
     \qquad & \text{(by Prop.~\ref{p:Adiff})} \\
  &= \sum_{i=1}^{k-1} \left(d_{k,i} \sum_{\!q\in D_k(n)\!} r_{k,i}^q\right). &
\end{align*}
\end{proof}

For $k\ge 6$, we now establish that $\delta_k$ cannot diverge
at a faster pace than the behavior seen in
Theorem~\ref{t:supdelta}.

\begin{proposition}\label{p:deltaO}
For $k\ge 6$, $|\delta_k(n)| = O(n^a)$ where
$a=\ln |r_{k,1}| /\ln \beta_k$.
\end{proposition}
\begin{proof}
We exploit Proposition~\ref{p:deltadecomp} and the fact that
$|r_{k,i}| \le |r_{k,1}|$ for all $1\le i< k$:
\[
|\delta_k(n)| \le
\sum_{q\in D_k(n)} \sum_{i=1}^{k-1} |d_{k,i}|\,|r_{k,i}|^q \le
M\!\!\sum_{q\in D_k(n)} \!\!|r_{k,1}|^q
\]
where $M = \sum_{1\le i} |d_{k,i}|$.
For $n\neq 0$, let us call $m$ the largest element in $D_k(n)$.
Since $\Sigma_k(D_k(n))=n$, we have $\A{k}{m}\le n$, hence by
Proposition~\ref{p:betaA} $\beta_k^m \le n$, so $|r_{k,1}|^m \le n^a$.
We over-estimate $D_k(n)$ as $\{0,1,\ldots,m\}$ and recall that
$|r_{k,1}|>1$ for $k\ge 6$:
\[
|\delta_k(n)| \le M\,\sum_{q=0}^m |r_{k,1}|^q =
M\, \frac{|r_{k,1}|^{m+1}-1}{|r_{k,1}|-1} \le C\,n^a
\]
where $C = M\,|r_{k,1}|\,(|r_{k,1}|-1)^{-1}$.
\end{proof}

\subsection{The case \texorpdfstring{$k=5$}{\it k=5}}

When $k=5$, the discrepancy $\Delta_5$ is still infinite, but the proof
is different from the general case $k\ge 6$ since here $|r_{5,1}|=1$ (actually
$r_{5,1} = e^{i\pi/3}$) and the divergence is quite slower.

\begin{theorem}\label{t:supdelta5}
$\Delta_5 = +\infty$.
More precisely, if we consider the sequences
\[
  u_n = \sum\limits_{p=0}^{n-1} \A{5}{6p} \qquad \text{and}\qquad
  u'_n = \sum\limits_{p=0}^{n-1} \A{5}{6p+3}
\]
 there exist some constants
$C,C_2,C'_2,K,K_2,K'_2\in\mathbb{R}_+^*$ such that for
 all $n\ge 0$,
\[
\delta_5(u_n) > C\,n - C_2 \qquad \text{and}\qquad
\delta_5(u_n) > K\,\ln(u_n) - K_2
\]
and
\[
\delta_5(u'_n) < -C\,n + C'_2 \qquad \text{and}\qquad
\delta_5(u_n) < -K\,\ln(u_n) + K'_2.
\]
\end{theorem}
\begin{proof}
The 5-decomposition
$D=\{0,6,\ldots,6(n{-}1)\}$ is canonical, since the positions $6p$ are
indeed apart by at least 5. Since the sum of this 5-decomposition is $u_n$,
then $D = D_5(u_n)$ by Theorem~\ref{t:zeck}.
Also note that here $r_{5,1} = e^{i\pi/3} = \overline{r_{5,2}}$ hence
$r_{5,1}^{6p} = 1 = r_{5,2}^{6p}$.
Moreover $d_{5,2} = \overline{d_{5,1}}$ and
$d_{5,4} = \overline{d_{5,3}}$ and $r_{5,4} = \overline{r_{5,3}}$
and $|r_{5,3}|<1$. So:

\begin{align*}
\delta_5(u_n) &= \sum_{i=1}^{4} \left(d_{5,i}\sum_{\!q\in D_5(u_n)\!}
r_{5,i}^q\right)
\qquad & \text{(by Prop.~\ref{p:deltadecomp})}\\
&=\sum_{i=1}^4 \left(d_{5,i}\sum_{p=0}^{n-1} r_{5,i}^{6p}\right) \\
&=2n\,\RE(d_{5,1}) + 2\,\RE\left(d_{5,3}\sum_{p=0}^{n-1}
r_{5,3}^{6p}\right).
\end{align*}

From the definition of $d_{5,1}$, we get that
$d_{5,1} \approx 0.0189 + 0.196 i$ and in particular its real part can
be proved to be strictly positive. We can hence choose
$C=2\,\RE(d_{5,1})$. We now give an upper bound for the last part:
\[
\left|2\,\RE\left(d_{5,3}\sum_{p=0}^{n-1} r_{5,3}^{6p}\right)\right| \le
2\,|d_{5,3}|\,\sum_{p=0}^{n-1} |r_{5,3}|^{6p} <
2\,|d_{5,3}|\,\frac{1}{1-|r_{5,3}|^6}
\]
so $2|d_{5,3}|(1-|r_{5,3}|^6)^{-1}$ is a suitable constant $C_2$
such that $\delta_5(u_n) > C\,n - C_2$.

Now, thanks to Proposition~\ref{Alincomb}:
\begin{align*}
u_n &= \sum_{p=0}^{n-1}\sum_{i=0}^4 c_{5,i}\,r_{5,i}^{6p} \\
    &= \sum_{i=0}^4 c_{5,i} \left(\sum_{p=0}^{n-1} r_{5,i}^{6p}\right) \\
    &= c_{5,0}\cdot \frac{\beta_5^{6n}-1}{\beta_5-1}
       + 2n\,\RE(c_{5,1}) +
       2\,\RE\left(c_{5,3}\cdot\frac{1-r_{5,3}^{6n}}{1-r_{5,3}^6}\right) \\
    &= K_0\,\beta_5^{6n} + O(n)
\end{align*}
where $K_0 = c_{5,0}(\beta_5-1)^{-1}$. Note that $c_{5,0} > 1$ and
$1 < \beta_5 < 2$ hence $K_0 > 1$. So there exists $N$ such that
for all $n>N$ we have
\begin{align*}
(K_0/2)\,\beta_5^{6n} &< u_n < (2K_0)\,\beta_5^{6n} \\
6n\,\ln \beta_5 + \ln(K_0/2) &< \ln(u_n) < 6n\,\ln \beta_5 +
\ln(2K_0) \\
C\,n + K\ln(K_0/2) &< K\ln(u_n) < C\,n + K\ln(2K_0)
\end{align*}
where we noted $K=C\,(6\ln \beta_5)^{-1}>0$.
Since $-K\ln(2K_0) \le K\ln(K_0/2)$, we obtain, still for $n>N$:
\[
|C\,n - K\ln(u_n)| < K\ln(2K_0).
\]
Combining this with the first part of this proof leads to:
\[
\delta_5(u_n) > K\ln(u_n) - C_2 - K\ln(2K_0)
\]
when $n>N$.
A suitable constant $K_2$ is hence the maximum of $C_2+ K\ln(2K_0)$ and of
all the initial distances $(\delta_5(u_n)-K\ln(u_n))$ for $0\le n \le N$.

The proof concerning $u'_n$ is similar, except that this time
$r_{5,1}^{6p+3}=-1$ and hence
\[
\delta(u'_n) = -2n\,\RE(d_{5,1}) + 2\,\RE\left(d_{5,3} r_{5,3}^3\sum_{p=0}^{n-1}
r_{5,3}^{6p}\right).
\]
The same constants $C$ and $K$ can be reused, while $C'_2$ and $K'_2$
need to be adapted, for instance
$C'_2=2|d_{5,3}\,r_{5,3}^3|(1-|r_{5,3}|^6)^{-1}$.
\end{proof}

This logarithmic behavior in $k=5$ cannot be outpaced:

\begin{proposition}\label{p:deltaO5}
$|\delta_5(n)| = O(\ln n)$.
\end{proposition}
\begin{proof}
Here, $|r_{5,1}|=|r_{5,2}|=1$ and $|r_{5,3}|=|r_{5,4}|<1$.
Thanks to Proposition~\ref{p:deltadecomp}:
\[
|\delta_5(n)| \le
\sum_{q\in D_5(n)} \sum_{i=1}^4 |d_{5,i}|\,|r_{5,i}|^q \le
M\!\!\sum_{q\in D_5(n)} \!\!1
\]
where $M = \sum_{1\le i} |d_{5,i}|$.
For $n\neq 0$, let us call $p$ the largest element in $D_5(n)$.
The number of elements in $D_5(n)$ is at most $p+1$.
Moreover $\A{5}{p}\le n$, hence by
Proposition~\ref{p:betaA} $\beta_5^p \le n$, so
$p{+}1 \le 1 + \ln n / \ln \beta_5$, and finally
$|\delta_5(n)| \le C\,\ln n + M$ with $C= M / \ln \beta_5$, hence
the desired asymptotic bound.
\end{proof}

\subsection{ Finite discrepancy up to \texorpdfstring{$k=4$}{\it k=4}}

First, the discrepancies $\Delta_1$ and $\Delta_2$ could easily be determined:
\begin{itemize}
\item From $F_1(n)=\lfloor\tfrac{n+1}{2}\rfloor$ and
$\alpha_1 = \tfrac{1}{2}$, we obtain
  $\delta_1(n) \in \{0,\tfrac{1}{2}\}$ for all $n\ge 0$ and
  then $\Delta_1 = \tfrac{1}{2}$.
\item From $F_2(n)=G(n)=\lfloor (n{+}1)/\varphi \rfloor$ and
  $\alpha_2 = \tfrac{1}{\varphi}=\varphi-1$ where $\varphi$ is the
  Golden Ratio, we obtain $\Delta_2 = \varphi{-}1$.
\end{itemize}

For precise estimates of $\Delta_3$ and $\Delta_4$, we first show
how to compute efficiently extrema of $\delta_k(n)$ when $0\le n< N$
and $N$ is an $\A{k}{p}$ number.
In this case, instead of computing $\delta_k(n)$ everywhere, we can
focus only on $\A{k}{p}$ numbers.
This part works for any $k\ge 1$, but will be used only for $k=3,4$.

\begin{definition}
For any integer $k\ge 1$, we define the following sequences of real numbers:
\[
\dmax{k}{p} =
\begin{cases}
 0, & \text{if $p=0$;} \\
 \max(\dmax{k}{p-1},\ \dmax{k}{p\ominus k} + \delta_k(\A{k}{p-1})), &
 \text{if $p>0$}
\end{cases}
\]
and
\[
\dmin{k}{p} =
\begin{cases}
 0, & \text{if $p=0$;} \\
 \min(\dmin{k}{p-1},\ \dmin{k}{p\ominus k} + \delta_k(\A{k}{p-1})), &
 \text{if $p>0$.}
\end{cases}
\]
\end{definition}

\begin{proposition}\label{p:maxmax}
Let $k \ge 1$ and $p\ge 0$. We abbreviate here $N = \A{k}{p}$. Then:
\[
\max_{n<N} \delta_k(n) = \dmax{k}{p}  \qquad \text{and} \qquad
\min_{n<N} \delta_k(n) = \dmin{k}{p}.
\]
Hence $\lim_{p\to+\infty} \dmax{k}{p} = \sup_n \delta_k(n)$
and $\lim_{p\to+\infty} \dmin{k}{p} = \inf_n \delta_k(n)$.
\end{proposition}
\begin{proof}
We proceed by induction on $p$.
The case $p=0$ is obvious since $N=1$ and hence $n=0$ and
$\delta_k(0) = F_k(0) - \alpha_k\,0 = 0$.
Now, for $p>0$, let us abbreviate $N' = \A{k}{p-1}$.
For a number $n<N$, either $n<N'$ or $N'\le n < N$. The latter
corresponds exactly to the situations where $n-N' < \A{k}{p\ominus k}$,
thanks to Equation~\eqref{Agenrec}. In these cases,
$\delta_k(n) = \delta_k(n-N') + \delta_k(N')$,
thanks to Theorem~\ref{t:Fshift}.
Finally,
\begin{align*}
\max_{n<N} \delta_k(n) &= \max(\max_{n<N'} \delta_k(n),
                                \max_{N'\le n<N} \delta_k(n)) \\
  &= \max(\max_{n<N'} \delta_k(n),
          \max_{n'<\A{k}{p\ominus k}} \delta_k(n') + \delta_k(N')) \\
  &= \max(\dmax{k}{p-1}, \dmax{k}{p\ominus k} + \delta_k(N'))
\qquad \mbox{(by I.H.)}\\
  &= \dmax{k}{p}
\end{align*}
From that, it is easy to see that the non-decreasing sequence
$(\dmax{k}{p})$
will eventually be greater or equal to any $\delta_k(n)$, while
being no more than a $\max \delta_k(n)$ at each step. Hence its limit
exists and is $\sup_n \delta_k(n)$.
The proofs for $\dmin{k}{p}$ are similar.
\end{proof}

We now estimate the remaining distance between $(\dmax{k}{p})$
and its limit and the same for $(\dmin{k}{p})$. This will
now require $|r_{k,i}|<1$ for all secondary roots ($i\neq 0$), hence
$2 \le k \le 4$.

\begin{proposition}\label{p:residue}
Let $2\le k \le 4$ and $p\ge 0$.
Then $\sup_n \delta_k(n) \le \dmax{k}{p} + R_k(p)$ and
$\inf_n \delta_k(n) \ge \dmin{k}{p} - R_k(p)$, where
\[
R_k(p) := \sum_{i=1}^{k-1} |d_{k,i}| \frac {|r_{k,i}|^p}{1-|r_{k,i}|^k}.
\]
\end{proposition}
\begin{proof}
Let $n\ge 0$. We consider its $k$-decomposition $D_k(n)$ (see
Theorem~\ref{t:zeck}) and partition it via comparisons with $p$:
\[
n = \sum_{q\in D_k(n)} \A{k}{q} =
\sum_{\substack{q\in D_k(n) \\ q<p}} \A{k}{q} +
\sum_{\substack{q\in D_k(n) \\ q\ge p}} \A{k}{q}
\]
Let us call $n_1$ and $n_2$ these two last sums.
Via Theorem~\ref{t:Fshift}, we can show that
$\delta_k(n) = \delta_k(n_1) + \delta_k(n_2)$.
Since $n_1 < \A{k}{p}$ (see for instance the proof of Theorem~\ref{t:zeck}),
we have $\dmin{k}{p} \le \delta_k(n_1) \le \dmax{k}{p}$ thanks to the
previous Proposition~\ref{p:maxmax}.
Moreover:
\[
|\delta_k(n_2)| \le
\sum_{i=1}^{k-1} \left(|d_{k,i}|\sum_{\!q\in D_k(n_2)\!}
|r_{k,i}|^q\right)
\le
\sum_{i=1}^{k-1} \left(|d_{k,i}|\sum_{m=0}^{+\infty} |r_{k,i}|^{p+km}\right)
\le R_k(p).
\]
Indeed, the first inequality comes from
Proposition~\ref{p:deltadecomp}.
Then the indices in $D_k(n_2)$ are all at least $p$, and
apart by at least $k$, so the $m$-th such index is at least $p{+}km$.
And the Proposition~\ref{p:secondary} implies here that
$|r_{k,i}| \le |r_{k,1}|<1$ for all $1 \le i < k$.

Combined together, this provides
$\dmin{k}{p} - R_k(p) \le \delta_k(n) \le \dmax{k}{p} + R_k(p)$
for all $n\ge 0$, hence the desired inequalities with $\sup$ and $\inf$.
\end{proof}

We now provide approximations of $\Delta_3$ with more than
32 correct decimal digits and $\Delta_4$ with more than 15 correct
digits.

\begin{theorem}\label{t:delta34}
\begin{align*}
d_3^+ - 3.10^{-33} \le \mbox{$\sup_n \delta_3(n)$} \le d_3^+ \quad
\text{with} \quad d_3^+ &= 0.854187179928304211983581540152668 \\
d_3^- \le \mbox{$\inf_n \delta_3(n)$} \le d_3^- + 3.10^{-33} \quad
\text{with} \quad d_3^- &= -0.708415898743967960305146324178773 \\
d_4^+ - 6.10^{-16} \le \mbox{$\sup_n \delta_4(n)$} \le d_4^+ \quad
\text{with} \quad d_4^+ &= 1.5834687793247475 \\
d_4^- \le \mbox{$\inf_n \delta_4(n)$} \le d_4^- + 6.10^{-16} \quad
\text{with} \quad d_4^- &= -1.5060895457389591.
\end{align*}
Hence $\Delta_3 = \sup_n \delta_3(n)$ and hence $d_3^+$ is also an
upper approximation of $\Delta_3$, with the same precision as above.
Similarly, $d_4^+$ is also an upper approximation of
$\Delta_4 = \sup_n \delta_4(n)$.
\end{theorem}
\begin{proof}
Thanks to Proposition~\ref{p:residue}, it suffices to find values of
$p$ for which the residues $R_k(p)$ are small enough, and then compute
$\dmax{k}{p}$ and $\dmin{k}{p}$ with the desired precision.
The full details of these computations can be found in the files
{\tt F3.v} and {\tt F4.v} of the Coq/Rocq development~\cite{LetouzeyCoqDev}.
In particular, all the inequalities mentioned here have been verified in
Coq, hence they cannot suffer from errors such as incorrect rounding.

For $k=3$ we use $p=400$ since it involves reasonable computations in
Coq, while ensuring $R_3(400) < 10^{-33}$.
Indeed here $r_{3,2} = \overline{r_{3,1}}$ and
$d_{3,2} = \overline{d_{3,1}}$, moreover $|r_{3,1}|^2 = \alpha_3 < 0.6824$
since $r_{3,0}\,r_{3,1}\,r_{3,2} = 1$ and $r_{3,0} = \beta_3 =
\alpha_3^{-1}$.
After some more computations we obtain
\[
R_3(2p) = \frac{2|d_{3,1}|}{1-|r_{3,1}|^3} \, |r_{3,1}|^{2p}
< 1.112 \cdot 0.6824^p
\]
which leads to $R_3(400) < 10^{-33}$.

For $k=4$ we use $p=600$, which is still tractable in Coq, but only
ensures $R_4(600) < 3.10^{-16}$, since the convergence is quite
slower, here $|r_{4,1}| \approx 0.9404$ and $|r_{4,1}|^2 \approx 0.8844$.

Now, for computing $\dmax{k}{p}$ and $\dmin{k}{p}$,
we exploit the fact that the numbers in the $(\dmax{k}{p})$
and $(\dmin{k}{p})$ sequences are always of the form $a-\alpha_k\,b$
with $a,b\in\mathbb{N}$. Under this form, it is even possible to
compare two such numbers in an exact manner. Indeed,
$a-\alpha_k\,b < a'-\alpha_k\,b' $ whenever $(a-a') < (b-b')\,\alpha_k$.
Since $\alpha_k$ is the unique positive root of $P_k(X)=X^k+X-1$,
the sign of $P_k(\tfrac{a-a'}{b-b'})$ will hence govern the result of the
comparison, except in trivial cases where the comparison result is
already clear, for instance when $b=b'$ or when
$(a-a')$ and $(b-b')$ are of opposite signs.
This way, the desired $\dmax{k}{p}$ and $\dmin{k}{p}$
are computed as pairs of exact integers, namely:
\begin{align*}
\dmax{3}{400} =\ &
 2031786811214411359348883471336991724172972024370943840592871475504 \,-\\
 & 2977728299822475173916958459765758872136894523385938812610760693222
 \,\alpha_3
 \\
\dmin{3}{400} =\ &
 1020161268160344624669178328493016309710214886667706164633381972074 \,-\\
 &1495119006490722158917214418259808655182461295204695261413364767438
 \,\alpha_3 \\
\dmax{4}{600} =\ &
 474542795998615222029347282639059927656268169929641003315959991602098 \\
 & 421282890067492 \,-\,
 655000776893753621409603547449877928720169482765627 \\
 & 175295567680977505721573702352765 \, \alpha_4 \\
\dmin{4}{600} =\ &
 915037483574937370155779315529924955263716715472701264249494125085598 \\
 & 423291500577325 \,-\,
 126300571346201255737296305828053313770378412375106 \\
 & 8121378399064986894058134103876852 \, \alpha_4
\end{align*}
As a final step, these expressions are converted to rational intervals
bounding
$\dmax{k}{p}$ and $\dmin{k}{p}$, using some high-precision rational
intervals
bounding $\alpha_k$. Indeed, $\alpha_k$ is
straightforward to approximate, for instance via the Newton method, and
these approximations can easily be verified in Coq with just a
computation of the sign of $P_k(X)$ there. We use a precision of
$10^{-100}$ when approximating $\alpha_3$, leading to
precision of $2.10^{-33}$ on $\dmax{3}{400}$ and $\dmin{3}{400}$.
Together with the precision on $R_3(400)$, this provides the announced
interval of diameter $3.10^{-33}$.
Similarly for $k=4$, the $10^{-100}$ precision used on $\alpha_4$
provides a $3.10^{-16}$ precision on $\dmax{4}{600}$ and
$\dmin{4}{600}$ that adds up with the precision on $R_3(400)$ to
give the announced intervals.
\end{proof}

\subsection{Applications of these discrepancy results}

In the OEIS pages~\seqnum{A5206} and \seqnum{A5375} about $F_3$ and
$F_4$ \cite{OEIS},
Cloitre conjectured twenty year ago that
$F_3-\lfloor \alpha_3\,n\rfloor$ is either 0
or 1, and the same concerning $F_4-\lfloor \alpha_4\,n\rfloor$.
We can now prove this conjecture concerning $F_3$ and a corrected
version of the conjecture about $F_4$.

\begin{corollary}
For all $n\ge 0$, $F_3(n) -\lfloor \alpha_3\,n\rfloor \in\{0,1\}$
and $F_4(n) - \lfloor \alpha_4\,n\rfloor \in\{-1,0,1,2\}$.
All these situations are reached.
\end{corollary}
\begin{proof}
From the previous theorem, $F_3(n)-1 < \alpha_3\,n < F_3(n)+1$
hence the integer $\lfloor \alpha_3\,n \rfloor$ can only be $F_3(n)-1$
or $F_3(n)$.
We handle $F_4$ similarly.
Now, $F_3(0)=0=\lfloor 0\,\alpha_3 \rfloor$ and
$F_3(1)=1=1+\lfloor 1\,\alpha_3 \rfloor$.
For $F_4(n)$, taking $n=0$ and $n=1$ also exhibits a difference of $0$
and $1$ with $\lfloor \alpha_4\,n\rfloor$. Differences of $-1$ and $2$
are less frequent, their first occurrences are $n=120$ for
$2$ and $n=243$ for $-1$.
\end{proof}

We can also express that for $k\le 4$, the functions $F_k$ are close to
be additive maps.

\begin{definition} A function $f : \mathbb{N}\to\mathbb{N}$ is said
to be \emph{almost additive} when
\[ \sup_{(n,m)\in\mathbb{N}^2} |f(n{+}m)-f(n)-f(m)| < +\infty. \]
\end{definition}

Interestingly, A'Campo~\cite{acampo2003} called \emph{slopes} similar
functions on integers and showed that they could be used as an
alternative construction for $\mathbb{R}$.

\begin{corollary}
Both $F_3$ and $F_4$ are almost additive: for all $n,m\ge 0$ we have
\begin{align*}
|F_3(n{+}m)-F_3(n)-F_3(m)| &\le 2 \\
|F_4(n{+}m)-F_4(n)-F_4(m)| &\le 4.
\end{align*}
\end{corollary}
\begin{proof}
$|F_3(n{+}m)-F_3(n)-F_3(m)|$ can be reformulated as
$|\delta_3(n{+}m)-\delta_3(n)-\delta_3(m)|$ which is strictly less
  that 3 by Theorem~\ref{t:delta34}.
Hence this integer quantity is at most $2$.
Similarly, for $F_4$ a strict bound is $3 \delta_4 < 5$, hence the large
inequality with $4$.
\end{proof}
For $F_3$, this almost-additivity bound of $2$ is reached:
for instance $F_3(18+78)=65$ while $F_3(18)+F_3(78)=13+54=67$ or
$F_3(39+168)=142$ while $F_3(39)+F_3(168)=26+114=140$.
But the bound $4$ for $F_4$ does not seem to be reached
and could probably be improved to $3$.

For similar reasons, $F_1$ and $F_2$ are also almost additive, with
bounds $1$. On the contrary, $F_k$ cannot be
almost additive when $k\ge 5$:

\begin{proposition}
For $k\ge 5$,
$\sup_{(n,m)\in\mathbb{N}^2} |F_k(n{+}m)-F_k(n)-F_k(m)| = +\infty$.
\end{proposition}
\begin{proof}
Let $k\ge 5$. Suppose there exists a constant $C$ such that
for all $n,m\ge 0$ we have $|F_k(n{+}m){-}F_k(n){-}F_k(m)| \le C$.
Then consider $u_n=F_k(n)+C$. This sequence is subadditive, since
$u_{n{+}m} = F_k(n{+}m)+C \le F_k(n)+F_k(m)+2C = u_n+u_m$.
Hence by Fekete's subadditive lemma, $\lim \tfrac{1}{n}u_n$ exists
and is equal to $\inf\tfrac{1}{n}u_n$. Now, from Theorem~\ref{t:limits}
we know that $\lim \tfrac{1}{n}u_n = \alpha_k$. Hence
$\inf \tfrac{1}{n}u_n = \alpha_k$ and so $u_n \ge \alpha_k\,n$ for all
$n$, i.e., $F_k(n)-\alpha_k\,n \ge -C$. Similarly, $v_n = F_k(n)-C$ is
superadditive and mutatis mutandis we also get $F_k(n)-\alpha_k\,n \le C$.
All in all, $|\delta_k(n)| < C$ for all $n$, which
contradicts either Theorem~\ref{t:supdelta} or
Theorem~\ref{t:supdelta5}.
\end{proof}

From Theorem~\ref{t:delta34} and the bounds on $\Delta_3$ and $\Delta_4$,
one may also derive
bounds on the difference between $F_3^j$ and its linear equivalent,
and similarly for $F_4^j$. For example:
\[
|F_3^2(n)-\alpha_3^2\,n| = |\delta_3(F_3(n))-\alpha_3\,\delta_3(n)|
\le (1+\alpha_3) \Delta_3 \le 1.4371.
\]
The bounds obtained this way are quite non-optimal, since the different
$\delta_k(F_k^j)$ at stake here are not visiting their ranges
independently.
In the particular case of $F_3^2$, it is possible to prove that
$-0.7864 \le |F_3^2(n)-\alpha_3^2\,n| \le 1.0393$
by considering first the numbers whose \mbox{3-decomposition} includes
the index 0 and then the other numbers.
Moreover $|F_3^2(n)-\alpha_3^2\,n|$ does slightly
exceed $1$ occasionally\footnote{In our experiments with $n$ up to $10^6$,
this happens about $0.1\%$ of the time.}. As a consequence,
$F_3^2(n) - \lfloor \alpha_3^2\,n \rfloor \in \{0,1,2\}$, and this
difference is rarely $2$, see for instance $n=1235$.

\begin{figure}[ht]
\includegraphics[width=\linewidth]{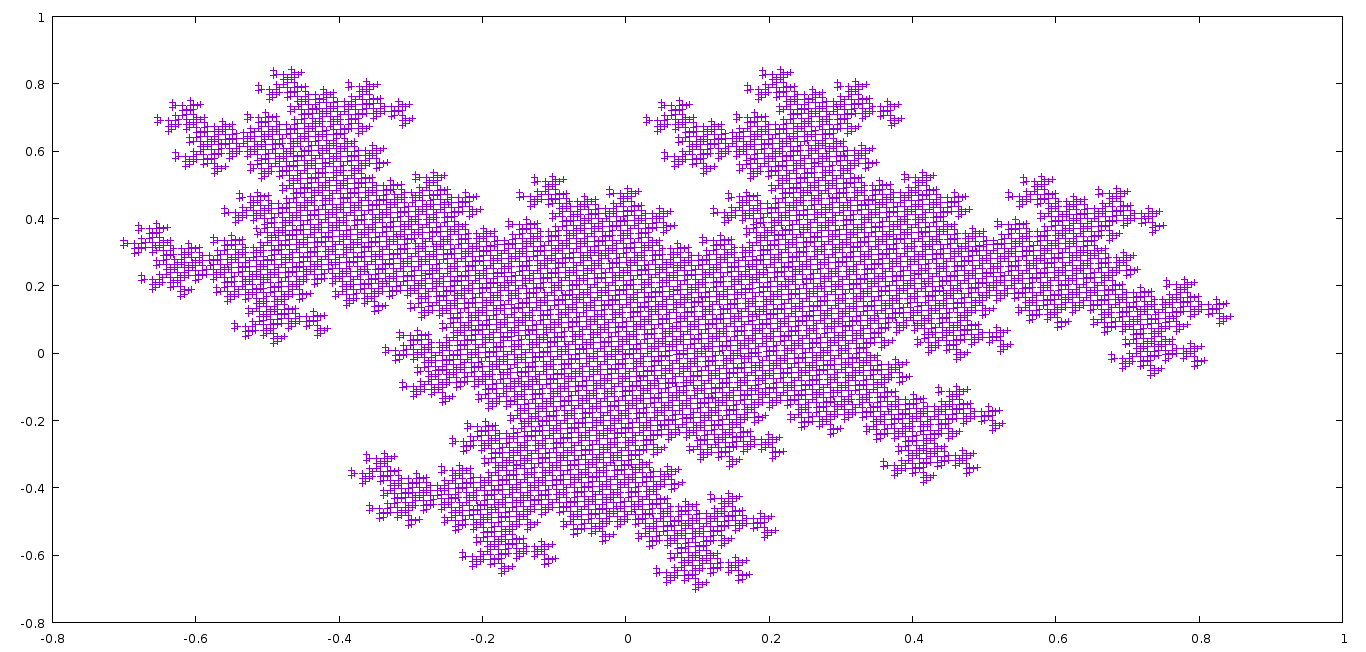}
\caption{A Jacobi-Perron fractal: displaying
  $(\delta_3(n),\delta_3(F_3(n)))$ for $n<10^4$.}
\label{f:fractal}
\end{figure}

During this study of $\delta_3$, the author made a
serendipitous encounter with the fractal shape displayed in
Figure~\ref{f:fractal}, when displaying the points
$(\delta_3(n),\delta_3(F_3(n)))$ (for example here for $n<10^4$).
Actually, this fractal happens to be the cover illustration of the
reference book named after Pytheas Fogg~\cite{Pytheas}; see in
particular its Chap. 8, where
this fractal is
associated with the modified Jacobi-Perron substitution $\sigma(1,0)$
(our $\tau_3$ up to a permutation of letters). More generally, it is
close to the Rauzy fractal~\cite{Rauzy82} obtained from the tribonacci
substitution. When giving a closer look, one may notice that our
figure is slightly distorted compared with the illustrations
in Pytheas Fogg~\cite{Pytheas}:
the constructions done there amount to consider
$(F_3(n)-\alpha_3\,n, F_3^2(n)-\alpha_3^2\,n)$ or some obvious symmetry
of it, and as we have just seen, this may be written as
$(\delta_3(n),\delta_3(F_3(n))+\alpha_3\,\delta_3(n))$.
So the distortion is actually the linear transformation
$(x,y)\mapsto(x,y\pm\alpha_3\,x)$. Either way, it is worth
noting that
such a spectacular fractal can be obtained in such a lightweight way,
via solely computing some images of the recursive function $F_3$
(i.e., Hofstadter's $H$)
as well as a constant $\alpha_3 \approx 0.6823278038280193$, zero of
$X^3{+}X{-}1$. In particular, no need for matrix reduction nor axis
projection.

As a future work, the case $k=4$ deserves some more investigation,
since it shares many aspects with the case $k=3$, in particular
$\beta_4 = 1/\alpha_4$ is also a Pisot number, and $\delta_4$ is also
bounded, with similar expression in terms of number decompositions.
Despite all that, our preliminary attempts lead to images that were
not so visually appealing, with rather ``smooth'' clouds of points.

As a final remark, recall that
$F_1(n)=\lfloor \tfrac{n+1}{2} \rfloor = \lceil\tfrac{n}{2} \rceil$
and $F_2(n) = \lfloor \tfrac{n+1}{\varphi} \rfloor$. We now prove
that no such expressions exist when $k\ge 3$: in this case $F_k$ cannot
be anymore the exact integer part of an affine function.
The case $k=3$ was already noted by Fine~\cite{Fine1986} and the rest
could be derived from Dilcher~\cite{Dilcher1993} but was not
explicitly stated there.

\begin{proposition}
\label{p:noexact}
For all $k\ge 3$ and $a,b\in\mathbb{R}$, there exists $n,n'\ge 0$
such that $F_k(n) \neq \lfloor an{+}b \rfloor$ and
$F_k(n') \neq \lceil an'{+}b \rceil$.
\end{proposition}
\begin{proof}
Suppose the existence of $a,b\in\mathbb{R}$ such that
$F_k(n) = \lfloor an{+}b \rfloor$ for all $n\ge 0$ or
$F_k(n) = \lceil an{+}b \rceil$ for all $n\ge 0$.
Since $\lim_{n\to\infty} \tfrac{1}{n} F_k = \alpha_k$,
the coefficient $a$ must be $\alpha_k$.
Then any of the above expressions would lead to an finite bound on
$|F_k(n)-\alpha_k n|$, which we proved impossible for $k\ge 5$.
Finally, for $k\in\{3,4\}$, we use some particular cases
to find a contradiction:
\begin{itemize}
\item $F_3(5)=4$ so $\lfloor 5\alpha_3 +b\rfloor=4$ would mean
$4\le 5\alpha_3 +b<5$ hence $b \ge 4-5\alpha_3 > 0.58$
\item $F_3(8)=5$ so $\lfloor 8\alpha_3 +b\rfloor=5$ would mean
$5\le 8\alpha_3 +b<6$ hence $b < 6-8\alpha_3 < 0.55$ which is
  incompatible with the previous constraint.
\item The same values of $F_3(5)$ and $F_3(8)$ also prevent the
existence of a real $b$ such that $F_3(n)=\lceil \alpha_3\,n+b \rceil$.
\item Similarly, $F_4(2)=1$ so $\lfloor 2\alpha_4 +b\rfloor=1$ would
  give $b < 0.56$ while $F_4(6)=5$ would give $b \ge 0.65$,
  contradiction. The same values also make it impossible to have
$F_4(n)=\lceil \alpha_4\,n+b\rceil$.
\end{itemize}
\end{proof}

\section{Conclusion}

The Figure~\ref{f:summary} provides a synthetic view of the results
concerning the functions~$F_k$.

\begin{figure}[ht]
\footnotesize
\renewcommand{\arraystretch}{1.2}
\begin{tabular}{|l|c|c|c|c|c|c|}

\hline
 & $F_1$ & $F_2$ & $F_3$ & $F_4$ & $F_5$ & $F_k$ for $k\ge 6$ \\
\hline
Hofstadter's name &   & $G$ & $H$ &   &   &  \\
\hline

Mean slope $\alpha_k$ & 0.5 & $\varphi{-}1$ & $\approx 0.682$ &
  $\approx 0.724$ & $\approx 0.754$ & root$(X^k{+}X{-}1)$ \\
\hline
Sup $|F_k(n)-\alpha_k n|$
  & $0.5$ & $\varphi{-}1$ & $<1$ & $<2$
  & $O(ln(n))$ & $O(n^a)$, $0{<}a{<}1$ \\
\hline

Exact expression  & $\lceil \frac{n}{2} \rceil
                     {=} \lfloor\frac{n+1}{2} \rfloor$
                  & $\lfloor \frac{n+1}{\varphi} \rfloor$
                  & \xmark & \xmark & \xmark & \xmark \\
\hline
Almost expression & & &
\parbox[c]{1.5cm}{\center $\lfloor \alpha_3 n \rfloor + \{0,1\}$}
& \parbox[c]{1.8cm}{\center $\lfloor \alpha_4 n \rfloor + \{-1,0,1,2\}$}
& \xmark & \xmark \\
\hline
Almost additive & \checkmark & \checkmark & \checkmark & \checkmark &
\xmark & \xmark \\

\hline
$\beta_k = \tfrac{1}{\alpha_k}$ is Pisot &
\checkmark & \checkmark & \checkmark & \checkmark &
$\checkmark^*$ & \xmark \\
\hline
\end{tabular}
\caption{Summary of results.}
\label{f:summary}
\end{figure}

Let us elaborate some more:
\begin{itemize}
\item
Here, by ``exact expression'', we mean exact integer part of an affine
function, and Proposition~\ref{p:noexact} disproved its existence for
$F_k$ when $k\ge 3$.
\item Similarly, an ``almost expression'' is here a finite
amount of possible differences from such an exact expression.
This cannot exists for $F_k$ when $k\ge 5$, since $|F_k(n)-\alpha_k n|$
is unbounded in this case.
\item
The above entry $\checkmark^*$ reminds that
$F_5$ is quite specific: even though $\beta_5$
is indeed a Pisot number (it is actually the smallest possible one,
i.e., the Plastic Ratio), the polynomial $Q_5 = X^5{-}X^4{-}1$
associated with $F_5$ is composed and admits two secondary zeros
of modulus 1, leading to an unbounded $|F_5-\alpha_5 n|$.
\end{itemize}

\section*{Acknowledgments}

I thank Wolfgang Steiner and Shuo Li and Jeffrey Shallit for their
helpful comments.

\appendix
\section{Coefficients of Proposition~\ref{Alincomb} seen as
  polynomial zeros}
\label{s:polycoeff}

\begin{proposition}
For any $k\ge 1$,
the coefficients $c_{k,i}$ seen in Proposition~\ref{Alincomb}
are the $k$
zeros of a polynomial in $\mathbb{Z}[X]$ of degree $k$, with leading
coefficient $k^k+(k{-}1)^{k-1}$ when $k>1$ (or $1$ when $k=1$) and
constant coefficient $(-1)$. This is also true for
the coefficients $\tilde{c}_{k,i}$ seen during the proof of
Proposition~\ref{Alincomb}.
\end{proposition}
\begin{proof}
We first prove the case of coefficients $\tilde{c}_{k,i}$.
For that, we consider the monic polynomial
$T = \prod_i (X-(\tilde{c}_{k,i})^{-1})$. Up to their signs, the
coefficients of $T$ are elementary symmetric
polynomials on
$((\tilde{c}_{k,i})^{-1})_{0\le i < k} = (Q'_k(r_{k,i}))_{0 \le i < k}$
and hence symmetric polynomials (in $\mathbb{Z}[X]$)
on $(r_{k,i})_{0 \le i < k}$. By the
fundamental theorem of symmetric polynomials, they can hence be
expressed as polynomials in $\mathbb{Z}[X]$ of the elementary symmetric
polynomials on $(r_{k,i})_{0\le i< k}$.
The latter are known to be integers (actually $0$ or $\pm 1$),
since $r_{k,i}$ are zeros of $Q_k$, a monic polynomial in $\mathbb{Z}[X]$.
This allows to
conclude that the polynomial $T$
belongs to $\mathbb{Z}[X]$, and so does its reciprocal
polynomial, whose zeros are all the $\tilde{c}_{k,i}$. Actually,
we rather consider here the opposite of the reciprocal of $T$,
hence its constant coefficient $(-1)$. Its leading coefficient is
the opposite of the constant coefficient of $T$, i.e.,
\begin{align*}
(-1)^{k-1}\prod_i (\tilde{c}_{k,i})^{-1}
&= (-1)^{k-1}\prod_i Q'_k(r_{k,i}) \\
&= (-1)^{k-1}\prod\limits_{i \neq j} (r_{k,i} - r_{k,j}) \\
&= (-1)^{\tfrac{1}{2}(k-1)(k-2)} \prod\limits_{i<j}
  (r_{k,i}-r_{k,j})^2 \\
&= (-1)^{\tfrac{1}{2}(k-1)(k-2)} {\text\it Disc}(Q_k)
\end{align*}
where ${\text\it Disc}(Q_k)$ is the discriminant of the polynomial $Q_k$.
This discriminant of $Q_k$ is also equal to
the discriminant of $P_k$, since $P_k$ is the opposite of the
reciprocal of $Q_k$. This discriminant can be proved to be
\[
(-1)^{\tfrac{1}{2}(k-1)(k-2)}.(k^k+(k-1)^{k-1})
\]
for $k>1$ (or $1$ when $k=1$).
This is a consequence of the work of Selmer~\cite{Selmer56}, it can
also be re-obtained by computing the adequate Sylvester matrix.
As a side note, this discriminant is also the square of the
determinant of the Vandermonde matrix for the $r_{k,i}$.
Finally, the desired leading coefficient is indeed
$k^k+(k-1)^{k-1}$ when $k>1$ and $1$ otherwise.

Now, for the coefficients $c_{k,i}$ we proceed similarly, except that
this time $(c_{k,i})^{-1}$ can be formulated as the value of an
integer polynomial evaluated at $(r_{k,i})^{-1}$:
\[
(c_{k,i})^{-1} = \frac{k\,r_{k,i}-(k-1)}{r_{k,i}^k} =
k ((r_{k,i})^{-1})^{k-1}-(k-1)((r_{k,i})^{-1})^k.
\]
Since $(r_{k,i})^{-1}$ is a zero of the polynomial $P_k(X)=X^k+X-1$,
we can conclude in the same way as before: the coefficients $c_{k,i}$
are the zeros of a polynomial in $\mathbb{Z}[X]$ with constant
coefficient $(-1)$. Finally the leading coefficient is the same
as before, since
\[
  \prod_i c_{k,i} = \prod_i (\tilde{c}_{k,i}\,r_{k,i}^{2k-2})
  = \left(\prod_i \tilde{c}_{k,i}\right)
    \left(\prod_i r_{k,i}\right)^{2(k-1)}
  = \prod_i \tilde{c}_{k,i}.
\]
Indeed, the product of the zeros $r_{k,i}$ are $\pm Q_k(0) = \pm 1$.
\end{proof}

As said earlier, we lack a general expression directly giving the
polynomials whose zeros are
the $c_{k,i}$ coefficients, or the $\tilde{c}_{k,i}$ ones. For the
moment, we can just provide such a polynomial for the terms
$(k\,r_{k,i}{-}(k{-}1))^{-1}$ appearing in both these coefficients.
But the method below does not
seem to extend further. We start by expressing the polynomial
whose $k$ zeros are the numbers $(k\,r_{k,i}{-}(k{-}1))$:
\begin{align*}
\prod_{i=0}^{k-1}(X - (k\,r_{k,i} {-} (k{-}1)))
 &= k^k \prod_{i=0}^{k-1} (\tfrac{X+k-1}{k} - r_{k,i}) \\
 &= k^k\, Q_k(\tfrac{X+k-1}{k}) \\
 &= (X+k{-}1)^k - k (X+k{-}1)^{k-1} - k^k \\
 &= X^k + \left(\sum_{j=1}^{k-1} \binom{k}{j} (j-1) (k{-}1)^{k-1-j} X^j \right)
    - k^k - (k{-}1)^{k-1}.
\end{align*}
In particular, this polynomial has a null coefficient of degree 1.
Now, by considering the reciprocal of this polynomial, we obtain that
the following polynomial admits
the numbers $(k\,r_{k,i}{-}(k{-}1))^{-1}$ as its zeros:
\[
(k^k - (k{-}1)^{k-1})\,X^k -
\left(\sum_{j=1}^{k-1} \binom{k}{j}(k-j-1)(k{-}1)^{j-1} X^j \right) - 1.
\]

\bibliographystyle{plain}
\bibliography{hofstadter}
\end{document}